\documentclass[english]{elsarticle}

\usepackage{lineno,hyperref}
\usepackage{array,booktabs}
\usepackage{epsfig}
\usepackage[english]{babel}
\usepackage{mathrsfs}
\usepackage{enumerate}
\usepackage{graphicx}
\usepackage{inputenc}
\usepackage{amsfonts,amsmath,amsxtra,amsthm,amssymb,latexsym}
\usepackage{caption}
\usepackage{color,colordvi}
\usepackage{enumerate}
\usepackage{rotating}

\newtheorem{theorem}{Theorem}[section]

\newtheorem{example}[theorem]{Example}

\newcommand{\N}{\mathbb{N}}
\newcommand{\R}{\mathbb{R}}

\newcommand{\e}{\mathrm{e}}
\newcommand{\CC}{\mathcal{C}}
\newcommand{\OO}{\mathcal{O}}

\begin{document}

\begin{frontmatter}

\title{Spectrum of periodic chain graphs with time-reversal non-invariant vertex coupling}

\author[1]{Marzieh Baradaran}
\ead{marzie.baradaran@yahoo.com}

\author[2,3]{Pavel Exner\corref{cor1}}
\ead{exner@ujf.cas.cz}
\cortext[cor1]{Corresponding author}

\author[3]{Milo\v{s} Tater}
\ead{tater@ujf.cas.cz}

\address[1]{Department of Mathematics, Faculty of Nuclear Sciences and Physical Engineering, Czech Technical University, B\v rehov\'a 7, 11519 Prague, Czechia}
\address[2]{Doppler Institute for Mathematical Physics and Applied Mathematics, Czech Technical University, B\v rehov\'a 7, 11519 Prague,
Czechia}
\address[3]{Department of Theoretical Physics, Nuclear Physics Institute, Czech Academy of Sciences, 25068 \v{R}e\v{z} near Prague, Czechia}

\begin{abstract}
We investigate spectral properties of quantum graphs in the form of a periodic chain of rings with a connecting link between each adjacent pair, assuming that wave functions at the vertices are matched through conditions manifestly non-invariant with respect to time reversal. We discuss, in particular, the high-energy behavior of such systems and the limiting situations when one of the edges in the elementary cell of such a graph shrinks to zero. The spectrum depends on the topology and geometry of the graph. The probability that an energy belongs to the spectrum takes three different values reflecting the vertex parities and mirror symmetry, and the band patterns are influenced by commensurability of graph edge lengths.
\end{abstract}

\begin{keyword}
\texttt{Quantum graph, periodic structure, time reversal non-invariance, spectral gaps}
\end{keyword}

\end{frontmatter}




\section{Introduction}
\setcounter{equation}{0}

One can say about quantum graphs that they belong to the class of ideas coming before their time. Proposed first more than eighty years ago \cite{Pa36} and briefly investigated \cite{RSch53} in line with the original proposal they were then forgotten for several decades and attracted attention only at the end of the 1980s in connection with the progress in semiconductor physics. In the years that followed they appeared to be a source of many deep questions going beyond this `second motivation'; we refer to the book \cite{BK13} for a thorough introduction to this field and a rich bibliography.

Among numerous applications of quantum graphs we single out a recent paper aiming at modeling the \emph{anomalous Hall effect} \cite{SK15} which served as an indirect motivation of the present work. Let us recall that in order to make the Hamiltonian of a quantum graph, which acts as a one-dimensional Schr\"odinger operator at each edge, a self-adjoint operator, one has to match the wave functions properly at the graph vertices. There is a number of way how to do that: the most general condition in a vertex $v$ connecting $n$ edges can be written in the form
\begin{equation}\label{genbc}
 (U-I)\psi(v)+i\ell(U+I)\psi'(v)=0,
\end{equation}
where $\psi(v)$ and $\psi'(v)$ are vectors of boundary values of the functions and their (outward) derivatives, $\ell>0$ is a parameter fixing the length scale, and $U$ is an $n\times n$ unitary matrix. This multitude raises the question about the meaning of the parameters since in general different matrices $U$ define a different physics. If the continuity of wave functions at the vertex is required, the number of parameters is significantly reduced. In this case we end up with the so-called $\delta$ coupling \cite{Ex96} which refers to $U= {2\over n+i\alpha}\mathcal{J}-I$ in \eqref{genbc} with the single parameter $\alpha\in\R$, where $\mathcal{J}$ is the matrix whose all entries are equal to one; the particular case with $\alpha=0$ is usually labeled as Kirchhoff. These couplings can be interpreted in the limiting sense starting from the dynamics in a family of thin tubes with Neumann boundary built around the graph `skeleton' that shrink transversally to zero width: such a limit leads to the Kirchhof coupling \cite{Po11, RSch53}, and the $\delta$ coupling is obtained by adding a properly scaled potential in the vertex region.

The $\delta$ coupling was also used in the above mentioned paper \cite{SK15}. However, in order to produce a Hall-type behavior, that is, a voltage perpendicular to the current passing though the sample without the presence of a magnetic field, the authors had to impose a preferential direction on the graph edges which is something hard to justify from the first principles. On the other hand, the family of condition \eqref{genbc} includes those which violate the time reversal invariance. A simple example proposed by two of us in \cite{ET18} corresponds to the matrix
\begin{equation}\label{U}
\quad U= {\scriptsize \left( \begin{array}{ccccccc}
0 & 1 & 0 & 0 & \cdots & 0 & 0 \\ 0 & 0 & 1 & 0 & \cdots & 0 & 0 \\ 0 & 0 & 0 & 1 & \cdots & 0 & 0 \\ \cdots & \cdots & \cdots & \cdots & \cdots & \cdots & \cdots \\ 0 & 0 & 0 & 0 & \cdots & 0 & 1 \\ 1 & 0 & 0 & 0 & \cdots & 0 & 0
\end{array} \right)}\,;
\end{equation}
it is chosen so that at a fixed value of the momentum, $k=\ell^{-1}$, the motion in the vertex is cyclic, cf.~\eqref{onshell} below. The lack of time reversal invariance is obvious having in mind that this operation is represented by complex conjugation.

Instead of following the mentioned motivation, however, our main focus will be another striking property of this coupling. It was noticed in \cite{ET18} that transport properties of such vertices depend on the graph topology, specifically on the vertex parity. This was first manifested on band spectra of infinite lattice graphs: comparing square and hexagonal lattice with the indicated coupling one finds that their spectra are at high energies dominated by bands and gaps, respectively \cite{ET18}. The same mechanism determines the high-energy behavior of finite periodic graphs \cite{EL19} and strip `waveguides' \cite{EL20}.

In another recent paper \cite{BET20} we investigated the spectral behavior of a one-dimensional periodic graph in the form of a ring chain of the form sketched in Fig.~\ref{Figure01}. We worked in the symmetric situation, $\ell_2=\ell_3=\pi$, and the length scale fixed by $\ell=1$, and found, in particular, that the spectrum is for any $\ell_1$ dominated by gaps the size of which grows with $k$, the square root of energy. The widths of spectral bands, on the other hand, are bounded being either asymptotically constant or decreasing as $\mathcal{O}(k^{-1})$ depending on whether the appropriate spectral condition has a double or a single root, respectively. What was important, however, was that all the gaps (in the positive part of the spectrum) closed as $\ell_1\to 0$, illustrating that this vanishing edge limit, which follows from the general result of \cite{BLS19}, can be rather non-uniform.

The danger of such examples is that they may be too particular giving results non-generic in a broader context. This motivates us examine a wider class of such chain graphs. Except for the obvious change of the spectrum associated with an overall scaling, we are going to work now with three parameters, $\ell_1,\,\ell_3$, and $\ell$, instead of one considered in \cite{BET20}. This analysis will provide deeper insights into the structure of the band spectrum and show, in particular, that its high-energy behavior can be far more complex than the two mentioned asymptotic types might indicate, especially in connection with the possible incommensurability of edge lengths involved. This concerns both the gaps patterns, in particular, the distribution of their widths, as well as the `conductivity'. It appear that gapless positive spectrum in case of direct coupling of the rings is closely related to the mirror symmetry of the graph; once it is violated the probability that an energy value is found in the spectrum for a graph with vertices of degree four is reduced to one half. This shows that an even degree of vertices of a periodic graph does not guarantee that its positive spectrum is dominated by bands. On the other hand, the same quantity is zero when the vertices are of degree three so the nonuniform character of the limit mentioned above is preserved.

Let us mention briefly the contents of the paper. In the next section we discuss the generic case when all the edge lengths are nonzero. We derive the spectral condition and analyze the band spectrum, in particular, we describe its behavior in the high-energy regime. We also find conditions under which some of the gaps close and show that the negative spectrum consists of at most two bands. The following two sections are devoted to the situations when one of the three edge length parameters is zero, and as a consequence, the chain vertices have an even parity.

\section{The general case with vertices of degree three}
\setcounter{equation}{0}

We begin the discussion with infinite chain graph of the form sketched in Fig.~\ref{Figure01} assuming that all the edge lengths involved are nonzero, $\ell_j>0$ for $j=1,2,3$, and as a consequence, all the graph vertices are of degree three. The operator to investigate is the Laplacian on the graph acting as $\psi_j\mapsto -\psi''_j$ on the $j$th edge. Its domain consists of functions which are locally $H^2$and to make the operator self-adjoint, we have to match them properly at the vertices.
\begin{figure}[h]
\centering
\includegraphics[scale=.3]{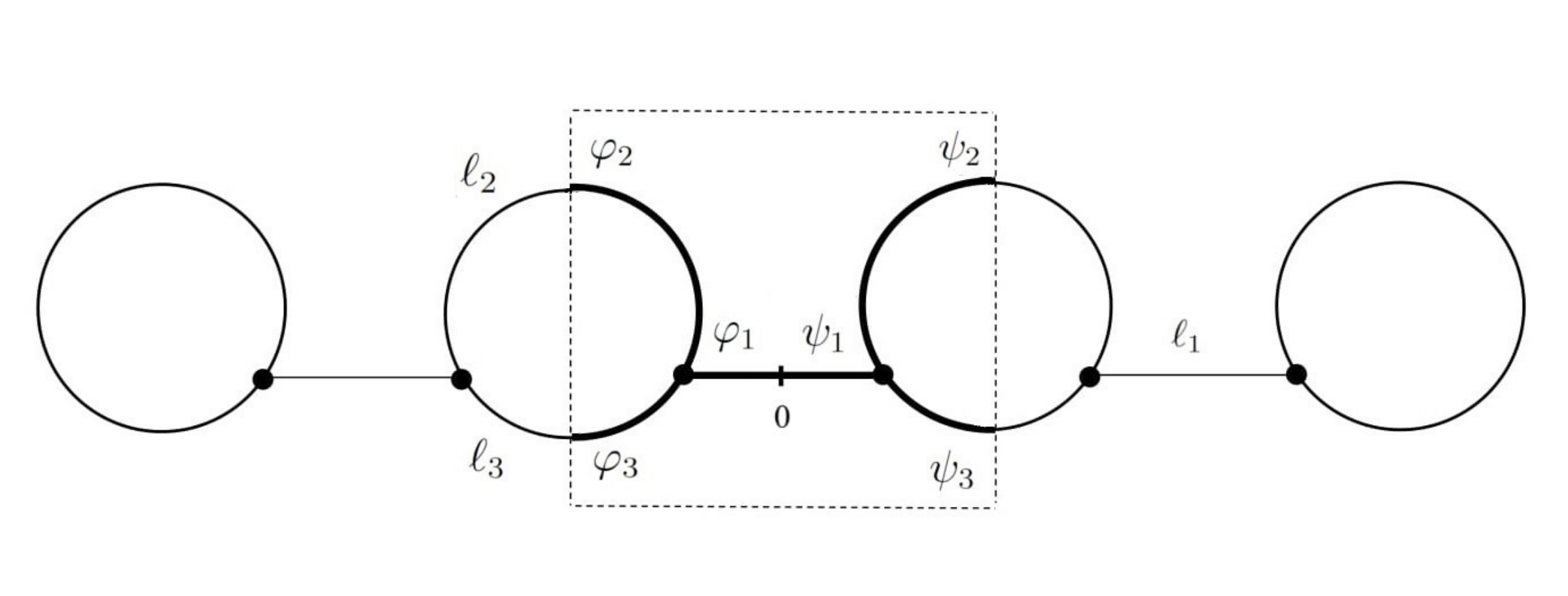}
\caption{ An elementary cell of the ring chain graph}
\label{Figure01}
\end{figure}

In the introduction we indicated our interest in a particular class of vertex couplings violating the time-reversal invariance which in the components read
\begin{equation}\label{coupling}
(\psi_{j+1}-\psi_{j})+i\ell(\psi_{j+1}^{\prime}+\psi_{j}^{\prime})=0, \quad\text{(cyclically)}
\end{equation}
where we use the symbols $\psi_j,\,\psi'_j$ with an abuse of notation for the boundary values of the function $\psi_j$ on the $j$th edge and its (outward) derivative. The condition contains another length-type parameter, $\ell$, which is a fixed positive number; its inverse is the momentum value at which the coupling exhibits the `maximum rotation' in the sense that the on-shell S matrix
\begin{equation}\label{onshell}
S(k) = \frac{k\ell-1 +(k\ell+1)U}{k\ell+1 +(k\ell-1)U}
\end{equation}
reduces for $k=\ell^{-1}$ to the cyclic matrix \eqref{U}, in the present case of the $3\times 3$ size. Since the spectral properties change in an obvious way when all the length-type quantities are scaled simultaneously, we may without loss of generality fix the scale requiring, for instance, that the upper and lower ring arcs lengths, $\ell_{2}$ and $\ell_{3}$, satisfy $\ell_{2}+\ell_{3}=2\pi$.

The chain graph we consider is periodic, hence the spectral analysis can be performed using the Floquet method \cite[Chap.~4]{BK13} writing the corresponding Hamiltonian as
\begin{equation}\label{Floquet1}
H = \int_{-\pi}^\pi H(\theta)\, \mathrm{d}\theta
\end{equation}
where the fiber $H(\theta)$ in the direct integral decomposition \eqref{Floquet1} acts on $L^2(\CC)$, where $\CC$ is the period cell and $\CC^*=[-\pi,\pi)$ is the dual cell, or Brillouin zone. Each of the operators $H(\theta)$ has a purely discrete spectrum and the spectrum of $H_\ell$ is the union $\bigcup_{\theta \in\CC^*} \sigma(H(\theta))$. It is well known that the unique continuation property does not hold in general in quantum graphs \cite[Sec.~3.4]{BK13}, hence $\sigma(H)$ may contain, in addition to the absolutely continuous part, infinitely degenerate eigenvalues, or `flat bands' as physicists usually call them; we will see that this indeed may happen here.

The elementary cell of our graph, indicated in Fig.~\ref{Figure01}, contains two vertices. Choosing the coordinates to increase from the left to right, we use for the wave function components the following Ansatz
\begin{align}\label{ansatz}
& \psi_{j}(x)=a_{j}^{+}\e^{ikx}+a_{j}^{-}\e^{-ikx},\quad x\in[0,\textstyle{\frac12}\ell_{j}],\nonumber \\
& \varphi_{j}(x)=b_{j}^{+}\e^{ikx}+b_{j}^{-}\e^{-ikx},\quad x\in[-\textstyle{\frac12}\ell_{j},0],
\end{align}
with $j=1,2,3$. The point $x=0$ is in the middle of the connecting edge, hence the function must be continuous there together with its derivative. Together with the Floquet conditions at the `free ends' of the cell we have
\begin{align}\label{Floquet}
&\psi _2\left(\textstyle{\frac12}{\ell_2}\right)-\e^{i \theta } \varphi _2\left(-\textstyle{\frac12}{\ell_2}\right)=0,\qquad\psi _2'\left(\textstyle{\frac12}{\ell_2}\right)-\e^{i \theta } \varphi _2'\left(-\textstyle{\frac12}{\ell_2}\right)=0,\nonumber \\
&\psi _3\left(\textstyle{\frac12}{\ell_3}\right)-\e^{i \theta } \varphi _3\left(-\textstyle{\frac12}{\ell_3}\right)=0,\qquad\psi _3'\left(\textstyle{\frac12}{\ell_3}\right)-\e^{i \theta } \varphi _3'\left(-\textstyle{\frac12}{\ell_3}\right)=0, \\
 &\psi _1(0)-\varphi _1(0)=0,\quad\quad\quad\quad\quad\quad\quad\psi _1'(0)-\varphi _1'(0)=0.\nonumber
\end{align}
This has to be complemented by the matching conditions (\ref{coupling}) at the vertices. Since the derivatives are taken in the outward direction, they read
\begin{align}\label{6eqs}
& \psi _3(0)-\psi _1\left(\textstyle{\frac12}{\ell_{1}}\right)+i \ell  \left(\psi _3'(0)-\psi _1'\left(\textstyle{\frac12}{\ell_{1}}\right)\right)=0,\nonumber \\
&\psi _2(0)-\psi _3(0)+i \ell  \left(\psi _2'(0)+\psi _3'(0)\right)=0,\nonumber \\
&\psi _1\left(\textstyle{\frac12}{\ell_{1}}\right)-\psi _2(0)+i \ell  \left(\psi _2'(0)-\psi _1'\left(\textstyle{\frac12}{\ell_{1}}\right)\right)=0,\\
&\varphi _2(0)-\varphi _1\left(-\textstyle{\frac12}{\ell_{1}}\right)+i \ell  \left(\varphi _1'\left(-\textstyle{\frac12}{\ell_{1}}\right)-\varphi _2'(0)\right)=0,\nonumber \\
&\varphi _3(0)-\varphi _2(0)+i \ell  \left(-\varphi _2'(0)-\varphi _3'(0)\right)=0,\nonumber \\
&\varphi _1\left(-\textstyle{\frac12}{\ell_{1}}\right)-\varphi _3(0)+i \ell  \left(\varphi _1'\left(-\textstyle{\frac12}{\ell_{1}}\right)-\varphi _3'(0)\right)=0.\nonumber
\end{align}
Inserting now from (\ref{ansatz}) into (\ref{6eqs}) and taking into account (\ref{Floquet}), we get a system of six linear equations for the coefficients $a_{j}^{\pm}$ and $b_{j}^{\pm}$, $j=1,2,3$. To be solvable, its determinant has to vanish; this requirement yields the spectral condition,
\begin{align} \label{Pos,SC,general,ell_1,2,3}
&2 \cos \theta  \big(k^2 \ell^2+1\big) \big(\sin k\ell_2 +\sin k\ell_3\big)+4 k \ell  \sin \theta \big(\cos k\ell_3 -\cos k\ell_2\big)\nonumber\\
&+2 \big(k^2 \ell^2+1\big) \big(\sin k\ell_1 -\sin k(\ell_1+\ell_3) \cos k\ell_2 -\sin k\ell_2 \cos k\ell_1 \cos k\ell_3\big)\nonumber\\
&+\big(k^4 \ell^4+3\big) \sin k\ell_1 \sin k\ell_2 \sin k\ell_3=0,
\end{align}
which can be simplified to the form
\begin{align}\label{Pos,SC,ell_1,3}
&\sin\pi k \big(16 \cos \theta \big(k^2 \ell^2+1\big)  \cos k \left(\pi -\ell_3\right)+32 k \ell  \sin \theta \sin k \left(\pi -\ell_3\right)\big)\nonumber\\
&+\big(k^2 \ell^2-1\big)^2 \big(\sin k \left(2 \pi -\ell_1\right)+2 \sin k\ell_1 \cos 2 k \left(\pi -\ell_3\right)\big)\nonumber\\
&+8 \big(k^2 \ell^2+1\big) \sin k\ell_1 -\big(k^2 \ell^2+3\big)^2 \sin k \left(\ell_1+2 \pi \right)=0,
\end{align}
taking into account that $\ell_2=2\pi-\ell_3$ holds by assumption. We divide the discussion of its implications into several parts.

\subsection{Positive spectrum}

Let us first investigate the flat bands. The condition \eqref{Pos,SC,ell_1,3} can have a solution independent of $\theta$ only if the first term on its left-hand side vanishes identically. Consider thus $k=n\in \mathbb{N}$, then (\ref{Pos,SC,ell_1,3}) reduces to
$$ 
-4 \big(n^2 \ell^2-1\big)^2 \sin n\ell_1 \sin^2 n\ell_3=0,
$$ 
and from this condition we infer that
\begin{itemize}
\item For $\ell=\frac{1}{n}$ and $k=n\in \mathbb{N}$, the number $k^{2}$ belongs to the spectrum independently of $\ell_{1}$ and $\ell_{3}$, being always embedded in the continuous spectrum.
\item Assuming that at least one of $\ell_{i},\:i=1,3$, is a rational multiple of $\pi$, in other words $\ell_{i}=\frac{p}{q}\pi$ with coprime $p,q\in\mathbb{N}$, the number $k^{2}=q^{2}n^{2}$ with $n\in\mathbb{N}$ belongs to the spectrum  for all $p$, independently of the other parameters, being always embedded in the continuous spectrum if the rationality concerns $\ell_3$ (or equivalently, $\ell_2$).
\item In particular, for $\ell_{1}=m\pi$ with $m\in\mathbb{N}$, the number $k^2=n^2$ belongs to the spectrum for any $n\in\mathbb{N}$, independently of the other parameters. Unless $\ell_3$ (or $\ell_2$) is simultaneously a multiple of $\pi$, these eigenvalues may not be embedded in the continuous spectrum.
\item While this concerns the degree-four case discussed in the subsequent sections, we note here also that if $\ell_{i}=0$ holds for at least one of $i=1,3,$ the number $k^2=n^2$ belongs to the spectrum for any $n\in\mathbb{N}$ independently of the other parameters, being always embedded in the continuous spectrum in the case of $\ell_3=0$ (or $\ell_2=0$).
\end{itemize}
The claims concerning the embedding follow from the continuous spectrum analysis to which we now proceed.

Away of these flat bands, the spectrum is (absolutely) continuous having a band-and-gap structure. To show it, we rewrite the condition  (\ref{Pos,SC,ell_1,3}) in the form
\begin{equation}\label{eq,abc,gen}
a\cos\theta + b\sin\theta =c
\end{equation}
with
\begin{align}\label{abc,gen}
&  a=    16 \big(k^2 \ell ^2+1\big) \sin k\pi \;\cos  k \left(\pi -\ell _3\right)     ,\nonumber \\
&  b=   32 k \ell  \;\sin k\pi \;\sin  k \left(\pi -\ell _3\right)    ,  \\
&  c=    -\big(k^2 \ell ^2-1\big)^2 \big(\sin  k(2 \pi -\ell _1)+2 \sin  k \ell _1  \cos  2 k (\pi -\ell _3) \big)   \nonumber\\
& \quad\; -8 \big(k^2 \ell ^2+1\big) \sin  k \ell _1 +\big(k^2 \ell ^2+3\big)^2 \sin  k(\ell _1+2 \pi)     .\nonumber
\end{align}
For a positive $k\not\in\N$ we have $a^2+b^2\neq 0$ and introducing $\sin\vartheta =\frac{a}{\sqrt{a^2+b^2}}$ and $\cos\vartheta =\frac{b}{\sqrt{a^2+b^2}}$, we can rewrite condition (\ref{eq,abc,gen}) as
\begin{equation}\label{eq,sin,al,th}
\sin (\vartheta +\theta )=\frac{c}{\sqrt{a^2+b^2}}.
\end{equation}
Consequently, $k^2$ belongs to a spectral band or gap, respectively, if
\begin{equation}\label{band,abc,gen,pos}
 a^2 + b^2-c^2 \geq 0,
\end{equation}
or
\begin{equation}\label{gap,abc,gen,pos}
 a^2 + b^2-c^2 <0.
\end{equation}

The band-gap structure depends on the parameters of the model. While generically the gaps are open, some of them may close; this typically happens if the neighboring band edges exhibit a crossing as illustrated in Fig.~\ref{Figure02}; note that these points appear in sequences having the same energy.
\begin{figure}[h]
\centering
\includegraphics[scale=.8]{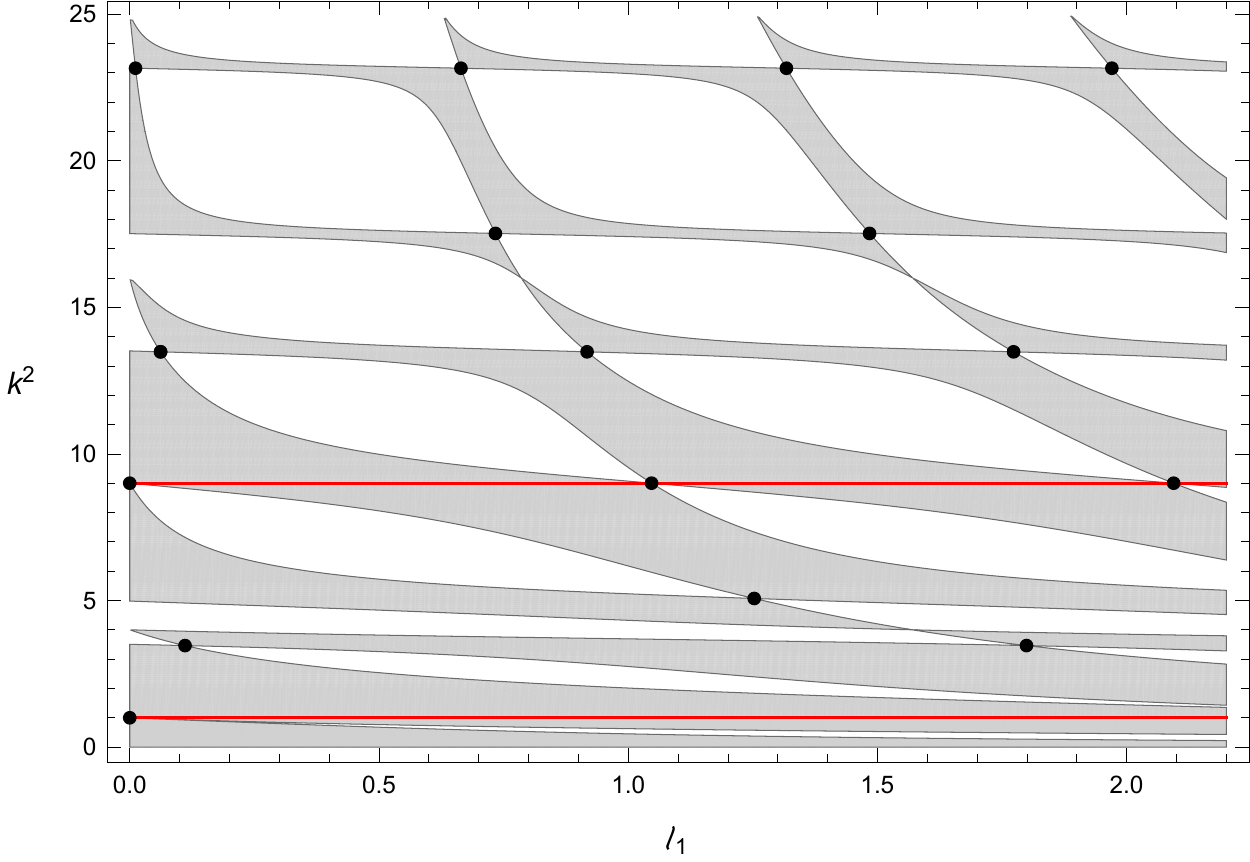}
\caption{The band edge crossings indicated by black dots for $\,\ell=1$, $\:\ell_3=\frac13\pi$ in dependence on $\ell_1>0$. The red lines correspond to flat bands indicated in the first two bullet points above, independent of $\ell_1$. On the other hand, the dots at $k^2=4$ and $16$ correspond to the flat bands of the second bullet point independent of $\ell_3$.}
\label{Figure02}
\end{figure}

We can identify such situations under commensurability requirements concerning relations of the lengths $\ell_3$ and $\ell_1$ and the ring perimeter. Let us focus on the former case, the latter can be dealt with in a similar way. If $\ell_3$ is a \emph{rational multiple} of $\pi$, \emph{i.e.} $\ell_{3}=\frac{m}{n}\pi$ with coprime $m,n\in\mathbb{N}$, the band edges cross at the points with the coordinates are $(\ell_1,k)=(\frac{j\ell_3}{im}, in),\;i,j\in\N\:$ (the case $\ell_3=\pi$, that is, $m=n=1$, is thus included). To see that we write the spectral condition as $\mathcal{F}(k,\ell_1,\theta)=0$, where $\mathcal{F}(k,\ell_1,\theta)$ is the left-hand side of (\ref{Pos,SC,ell_1,3}) for the particular values of $\ell_3$. Sufficient conditions to have such a crossings are
\begin{equation}\label{cross,con}
\frac{\partial \mathcal{F}(k,\ell_1,\theta)}{\partial \ell_1}=\frac{\partial \mathcal{F}(k,\ell_1,\theta)}{\partial \theta}=\frac{\partial \mathcal{F}(k,\ell_1,\theta)}{\partial k}=0\,;
\end{equation}
the first and last are obvious from Fig.~\ref{Figure02}, the middle one comes from the observation that the dispersion curves are smooth and the band edges correspond to their extrema. The derivatives can be easily calculated,
 \begin{align}
&\frac{\partial \mathcal{F}(k,\ell_1,\theta)}{\partial \ell_1}=-k \big(k^2 \ell ^2-1\big)^2 \left(\cos k \left(2 \pi -\ell _1\right)-2 \cos k\ell_1 \cos 2 k \left(\pi -\textstyle{\frac{\pi m}{n}}\right)\right) \nonumber \\
& \hspace{7em} +8 k \big(k^2 \ell ^2+1\big) \cos k\ell_1 -k \big(k^2 \ell ^2+3\big)^2 \cos k \left(\ell _1+2 \pi \right), \label{cross_ell_1}\\[.5em]
& \frac{\partial \mathcal{F}(k,\ell_1,\theta)}{\partial \theta}=32 k \ell\,  \cos \theta\,  \sin k \left(\pi -\textstyle{\frac{\pi m}{n}}\right)-16 \sin\theta  \big(k^2 \ell ^2+1\big) \cos k \left(\pi -\textstyle{\frac{\pi m}{n}}\right) \label{cross_theta}, \\[.5em]
& \frac{\partial \mathcal{F}(k,\ell_1,\theta)}{\partial k}= \frac{16}{n} \sin (\pi  k) \sin k \big(\pi -\textstyle{\frac{\pi m}{n}}\big) \big(\pi  \cos\theta\, \big(k^2 \ell ^2+1\big) (m-n)+2n\ell  \sin\theta \big)\nonumber \\
& \quad  +\big(k^2 \ell ^2-1\big)^2 \Big(2\ell_1 \cos k\ell_1\, \cos 2 k \big(\pi -\textstyle{\frac{\pi m}{n}}\big)-4 \big(\pi -\textstyle{\frac{\pi m}{n}}\big) \sin k\ell_1\, \sin 2 k \big(\pi -\textstyle{\frac{\pi m}{n}}\big)\Big)  \nonumber \\
& \quad   +4 k\ell^2 \big(k^2 \ell ^2-1\big) \Big(2 \sin k\ell_1\, \cos 2k \big(\pi -\textstyle{\frac{\pi m}{n}}\big)+\sin k(2\pi-\ell_1)\big)+16 k\ell^2 \sin k\ell_1  \nonumber \\
& \quad +\pi \cos\pi k \Big(16 \cos\theta\, \big(k^2 \ell ^2+1\big) \cos k \big(\pi -\textstyle{\frac{\pi m}{n}}\big)+32 k\ell \sin\theta\,  \sin k \big(\pi -\textstyle{\frac{\pi m}{n}}\big)\Big) \nonumber \\
& \quad    -4 k\ell^2 \big(k^2 \ell ^2+3\big) \sin k(\ell_1+2\pi)- (\ell_1+2\pi) \big(k^2 \ell ^2+3\big)^2 \cos k(\ell_1+2\pi) \nonumber \\
& \quad  +\frac{32}{n}\;k\ell  \sin\pi k\, \cos k\big(\pi -\textstyle{\frac{\pi m}{n}}\big) \big((n-m)\pi \sin\theta +n\ell \cos\theta \big) \nonumber \\
& \quad  +(2\pi -\ell_1) \big(k^2 \ell ^2-1\big)^2 \cos k(2\pi -\ell_1) +8\ell_1 \big(k^2 \ell ^2+1\big) \cos k\ell_1. \label{cross_k}
\end{align}
To solve the equations (\ref{cross,con}) we first note that (\ref{cross_ell_1}) is independent of $\theta$ and it can be easily checked that it vanishes for $k= i n$. On the other hand, substituting $k=i n$ into (\ref{cross_theta}) and (\ref{cross_k}), we obtain the conditions
$$ 
-16 \sin \theta  \left(i^2 n^2 \ell ^2+1\right) \cos\left(i(n-m)\pi\right)=0
$$ 
 and
$$ 
 16 \pi  \left(i^2 n^2 \ell ^2+1\right) \big(\cos \theta  \cos \left(i m\pi\right)-\cos \left(i n \ell _1\right)\big) =0,
$$ 
respectively. As can be seen, the former is satisfied only in the center of the Brillouin zone or at its edges, $\theta=0,\pm\pi$. Inspecting the latter at those points, we find that it is fulfilled for $\ell_1=\frac{j}{in}\pi$, that is
$$ 
16 \pi  \left(i^2 n^2 \ell ^2+1\right) \big(\cos \theta  \cos \left(i m\pi\right)-\cos \left(j\pi\right)\big) =0 .
$$ 
This condition can be always satisfied depending on the parity of $im$ and $j$: for the same and opposite parity, it happens with $\theta=0$ and $\theta=\pm\pi$, respectively.

Let us summarize the results about the positive spectrum obtained so far:
\begin{itemize}
\item \emph{Flat bands:} a single one, namely $k^2=n^2$, exists if the length parameter $\ell=\frac{1}{n}$ for some $n\in \mathbb{N}$. The existence of the other ones follows from the commensurability of the edges with the loop length: if $\ell_{i}=\frac{p}{q}\pi$ with coprime $p,q\in\mathbb{N}$, the number $k^{2}=q^{2}n^{2}$ is a flat band for any $n\in\mathbb{N}$, being moreover embedded in the continuum for $i=3$ but not necessarily so for $i=1$.
\item Away from the flat bands the spectrum is absolutely continuous having a bandgap structure.
\item Particular gaps may close under commensurability conditions: if $\ell_3=\frac{m}{n}\pi$ with coprime $m,n\in\mathbb{N}$ and $\ell_1=j\ell_3$ with $j\in\mathbb{N}$, neighboring bands touch at the energies $k^2=i^2n^2$ for any $i\in\mathbb{N}$.
\end{itemize}
The dependence of the spectrum on the parameter $\ell_3$ for a fixed length $\ell_1$ of the connecting links showing, in particular, the gap closing is illustrated in Figs.~\ref{Figure03} and \ref{Figure04}.
\begin{figure}[h]
  \centering
\includegraphics[width=0.8\textwidth]{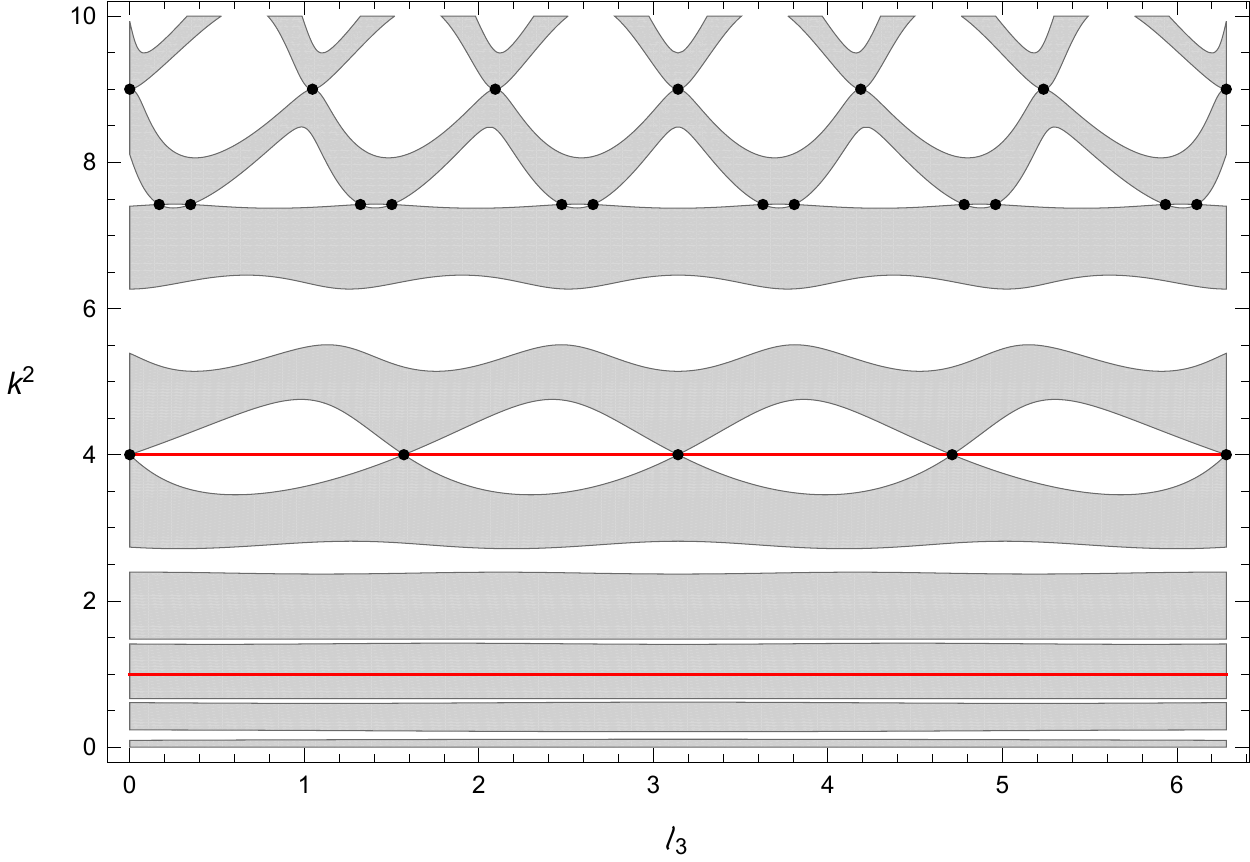}
  \caption{The positive spectrum for $\ell=1$ and $\ell_1=\frac32\pi$}
  \label{Figure03}
\end{figure}

\begin{figure}[h]
  \centering
\includegraphics[width=0.8\textwidth]{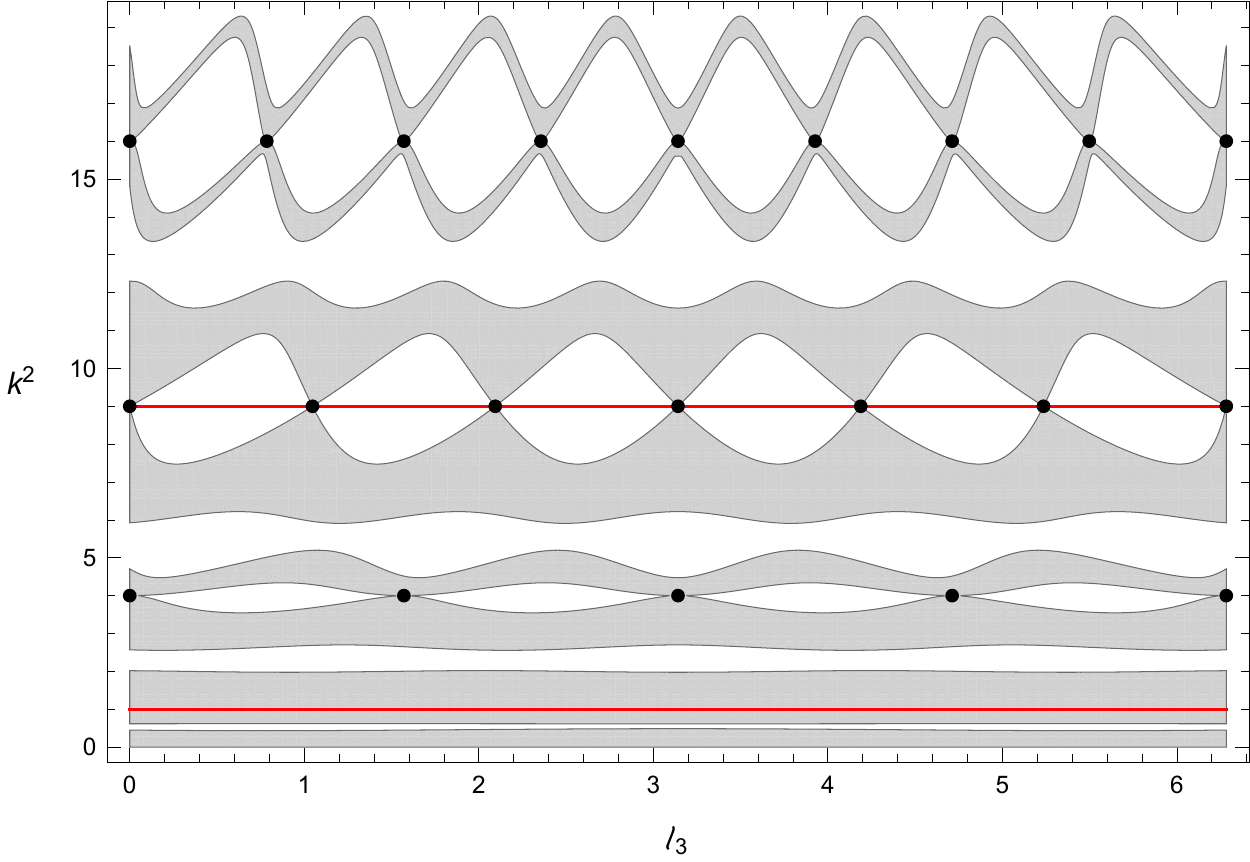}
  \caption{The positive spectrum for $\ell=1$ and $\ell_1=\frac13\pi$}
  \label{Figure04}
\end{figure}

\subsection{Asymptotic behavior of spectral bands}
\label{ss:asympt}

The structure of bands and gaps, in particular, their behavior at high energies are vital for transport properties of such a chain. In \cite{BET20} we have shown that for $\ell=1$, $\ell_2=\ell_3=\pi$, and any $\ell_1>0$ the spectrum is dominated by gaps and the band widths remain bounded; we also identified two types of behavior as $k\to\infty$, the asymptotically constant and decreasing as $\mathcal{O}(k^{-1})$. To get a deeper insight and to see how the spectrum looks like in our more general case, let us substitute from \eqref{abc,gen} into \eqref{gap,abc,gen,pos}; keeping the highest power of $k$, we get
\begin{equation} \label{asymptcond}
-16\;k^8 \ell ^8 \;\sin ^2 k \ell _1  \;\sin ^2 k \left(2 \pi -\ell _3\right) \; \sin ^2 k \ell _3 +\OO(k^6) <0 .
\end{equation}
Consequently, the spectrum is again dominated at high energies by gaps -- as expected for a chain with vertices of degree three -- because bands may exist only in the vicinity of the points
\begin{equation} \label{Sp,points}
k_{j,n}=\frac{n\pi}{\ell_{j}},\quad n\in\mathbb{N},
\end{equation}
at which the leading term vanishes.

Let us take a closer look at their structure. To this end, it is useful to rewrite the spectral condition (\ref{Pos,SC,general,ell_1,2,3}) in the form $A_{4}k^{4}+A_{2}k^{2}+A_{1}k+A_{0}=0$, where
\begin{align*} 
A_{4}:=&\, 4 \ell ^4 \sin  k \ell _1   \sin  k \ell _2  \sin  k \ell _3, \\
A_{2}:=&\, 2 \ell ^2 \left(\cos \theta  \big(\sin  k \ell _2 +\sin  k \ell _3 \big)+\sin k \ell _1 -\sin  k \ell _2  \cos  k \ell _1  \cos  k \ell _3 \right) \\
       &-2 \ell ^2 \sin k \left(\ell _1+\ell _3\right) \cos  k \ell _2 , \\
A_{1}:=&\, 4 k \ell  \sin \theta  \big(\cos  k \ell _3 -\cos  k \ell _2 \big), \\
A_{0}:=&\, 2 \Big(\sin  k\ell_2  \big(\cos\theta -\cos k\ell_1 \cos k\ell_3 \big) +\cos\theta \sin k\ell_3 -\sin k\left(\ell_1+\ell_3\right) \cos k\ell_2 \Big) \\
       &+\sin  k \ell _1  \big(3 \sin  k \ell _2  \sin  k \ell _3 +2\big).
\end{align*}
Since the cubic term is absent, the spectral condition acquires the asymptotic form
\begin{equation}\label{asycon}
A_{4}+\frac{A_{2}}{k^{2}}=\mathcal{O}(k^{-3}),
\end{equation}
from which we can deduce the behavior of bands in the vicinity of the points \eqref{Sp,points} at which the wavelength is commensurate with an edge.

The analysis in the three cases is similar, thus we present a detailed discussion in the first of them only and restrict ourselves to the results in the others. To describe the bands around $k=k_{1,n} = \frac{n\pi}{\ell_{1}}$, we set
$$ 
k=\frac{n\pi}{\ell_{1}}+\delta,\quad n\in\mathbb{N}.
$$ 
In that case we have $k^{-2}=\Big(\frac{\ell _1}{\pi  n}\Big)^2+\mathcal{O}(n^{-3})$ as $n\to\infty$, and the expansions
\begin{align*}
&\cos k \ell _3 =\cos \frac{\pi  n \ell _3}{\ell _1}-\delta  \ell _3 \sin \frac{\pi  n \ell _3}{\ell _1} -\frac{\delta ^2 \ell _3^2}{2}  \cos \frac{\pi  n \ell _3}{\ell _1} +\frac{\delta ^3 \ell _3^3}{6}  \sin  \frac{\pi  n \ell _3}{\ell _1}  +\mathcal{O}(\delta ^4), \\
&\sin k \ell_3 = \sin\frac{\pi  n \ell _3}{\ell _1} +\delta  \ell _3 \cos\frac{\pi  n \ell _3}{\ell _1} -\frac{\delta ^2 \ell _3^2}{2}  \sin\frac{\pi  n \ell _3}{\ell _1} -\frac{\delta ^3 \ell _3^3}{6}  \cos \frac{\pi  n \ell _3}{\ell _1} +\mathcal{O}(\delta ^4), \\
&\sin k \ell_1=  \delta  (-1)^n \ell _1-\frac{1}{6} \delta ^3 (-1)^n \ell _1^3        +\mathcal{O}(\delta ^4), \\
& \cos k \ell_1 =   (-1)^n-\frac{1}{2} \delta ^2 (-1)^n \ell _1^2        +\mathcal{O}(\delta ^4), \\
&\sin k \ell_2 = \sin\frac{\pi  n \ell _2}{\ell _1} +\delta  \ell _2 \cos\frac{\pi  n \ell _2}{\ell _1} -\frac{\delta ^2 \ell _2^2}{2}  \sin\frac{\pi  n \ell _2}{\ell _1} -\frac{\delta ^3 \ell _2^3}{6}  \cos \frac{\pi  n \ell _2}{\ell _1} +\mathcal{O}(\delta ^4), \\
& \cos k \ell_2 =  \cos\frac{\pi  n \ell _2}{\ell _1} -\delta  \ell _2 \sin\frac{\pi  n \ell _2}{\ell _1} -\frac{\delta ^2 \ell _2^2}{2}  \cos\frac{\pi  n \ell _2}{\ell _1} +\frac{\delta ^3 \ell _2^3}{6}  \sin \frac{\pi  n \ell _2}{\ell _1} +\mathcal{O}(\delta ^4),
\end{align*}
in terms of $\delta$. Substituting these relations in (\ref{asycon}), we see that $\delta$ has to satisfy the following relation,
\begin{align}
\frac{\beta _2 \cos \theta +\beta _1}{\pi ^2 n^2} &+\delta  \Big(\beta _3+\frac{\beta _5 \cos \theta +\beta _4}{\pi ^2 n^2}\Big)+\delta ^2 \Big(\beta _6+\frac{\beta _8 \cos \theta +\beta _7}{\pi ^2 n^2}\Big) \nonumber \\ & +\delta ^3 \Big(\beta _9+\frac{\beta _{10}+\beta _{11} \cos \theta }{\pi ^2 n^2}\Big) + \mathcal{O}(\delta^4)=0, \label{eq,beta,Asy}
\end{align}
where the coefficients $\beta _i$, $i=1,...,11$  are functions of the parameters $n,\ell,\ell_j$; it is straightforward if tedious to find them explicitly:
\begin{align}\label{beta,ell_1}
&\beta _1=   2 (-1)^{n+1} \ell ^2 \ell _1^2 \sin \frac{2 \pi ^2 n}{\ell _1}       , \qquad \beta _2=2 \ell ^2 \ell _1^2 \left(\sin \frac{\pi  n \ell _2}{\ell _1}+\sin \frac{\pi  n \ell _3}{\ell _1}\right),  \nonumber  \\[14pt]
&\beta _3=  (-1)^n \ell _1 \ell ^4 \sin \frac{\pi  n \ell _2}{\ell _1} \sin \frac{\pi  n \ell _3}{\ell _1}     ,   \nonumber  \\[14pt]
&\beta _4=   2 (-1)^{n+1} \ell ^2 \ell _1^2 \left(2 \pi  \cos \frac{2 \pi ^2 n}{\ell _1}+\ell _1 \left(\cos\frac{\pi  n \ell _2}{\ell _1} \cos \frac{\pi  n \ell _3}{\ell _1}-1\right)\right)   ,   \nonumber  \\[14pt]
&\beta _5=  2 \ell ^2 \ell _1^2 \left(\ell _2 \cos \frac{\pi  n \ell _2}{\ell _1}+\ell _3 \cos \frac{\pi  n \ell _3}{\ell _1}\right)      ,   \nonumber  \\[14pt]
&\beta _6=   (-1)^n \ell ^4 \ell _1 \left(\ell _2 \sin \frac{\pi  n \ell _3}{\ell _1} \cos \frac{\pi  n \ell _2}{\ell _1}+\ell _3 \sin \frac{\pi  n \ell _2}{\ell _1} \cos\frac{\pi  n \ell _3}{\ell _1}\right),    \\[14pt]
&\beta _7=    (-1)^n \ell ^2 \ell _1^2 \left(\ell _1^2 \sin \frac{2 \pi ^2 n}{\ell _1}+4 \pi ^2 \sin \frac{2 \pi ^2 n}{\ell _1}\right)    \nonumber  \\
&\quad\quad  +2 (-1)^n \ell ^2 \ell _1^3 \left(\ell _2 \sin \frac{\pi  n \ell _2}{\ell _1} \cos \frac{\pi  n \ell _3}{\ell _1}+\ell _3 \sin \frac{\pi  n \ell _3}{\ell _1} \cos \frac{\pi  n \ell _2}{\ell _1}\right)   ,\nonumber  \\[14pt]
&\beta _8= -\ell ^2 \ell _1^2 \left(\ell _2^2 \sin \frac{\pi  n \ell _2}{\ell _1}+\ell _3^2 \sin \frac{\pi  n \ell _3}{\ell _1}\right)     ,   \nonumber  \\[14pt]
&\beta _9=    \frac{1}{6} (-1)^n \ell ^4 \ell _1 \left(6 \ell _2 \ell _3 \cos \frac{\pi  n \ell _2}{\ell _1} \cos \frac{\pi  n \ell _3}{\ell _1}-\left(\ell _1^2+3 \left(\ell _2^2+\ell _3^2\right)\right) \sin \frac{\pi  n \ell _2}{\ell _1} \sin \frac{\pi  n \ell _3}{\ell _1}\right),      \nonumber   \\[14pt]
&\beta_{10}=   (-1)^n \ell ^2 \ell _1^3 \left(\left(\ell _2^2+\ell _3^2\right) \cos \frac{\pi  n \ell _2}{\ell _1} \cos \frac{\pi  n \ell _3}{\ell _1}-2 \ell _2 \ell _3 \sin \frac{\pi  n \ell _2}{\ell _1} \sin \frac{\pi  n \ell _3}{\ell _1}\right)         \nonumber  \\
&\quad\quad +\frac{1}{3} (-1)^n \ell ^2 \ell _1^2 \left(\ell _1^3 \left(\cos \frac{\pi  n \ell _2}{\ell _1}\cos \frac{\pi  n \ell _3}{\ell _1}-1\right)+6 \pi  \ell _1^2 \cos \frac{2 \pi ^2 n}{\ell _1}+8 \pi ^3 \cos \frac{2 \pi ^2 n}{\ell _1}\right) ,    \nonumber  \\[14pt]
&\beta_{11}= -\frac{1}{3} \ell ^2 \ell _1^2 \left(\ell _2^3 \cos \frac{\pi  n \ell _2}{\ell _1}+\ell _3^3 \cos \frac{\pi  n \ell _3}{\ell _1}\right)           \nonumber .
\end{align}
In the leading order, $\delta$ can be determined from the first two terms in Eq. (\ref{eq,beta,Asy}),
$$ 
\delta=\frac{-\beta _2 \cos\theta -\beta _1}{\beta _5 \cos\theta+\beta _4+\pi ^2 \beta _3 n^2},
$$ 
so that it behaves asymptotically as
$$ 
\delta=\frac{1}{\pi ^2 n^2}\;\frac{-\beta _2 \cos\theta-\beta _1}{\beta _3}+\mathcal{O}(n^{-4}),
$$ 
provided that $\beta_3$ which also depends on $n$ \emph{does not vanish and stays away from zero}. In that case we obtain the following general asymptotic expressions for energy and the width of the bands
$$ 
E_{n}(\theta)=k_\theta^2=\Big(\frac{\pi n}{\ell_1}\Big)^2+\frac{2}{\ell_1}\;\frac{1}{\pi n}\;\frac{-\beta _2 \cos\theta-\beta _1}{\beta _3}+\mathcal{O}(n^{-3}),
$$ 
\begin{equation} \label{bandwidth1}
\triangle E_{n}=k_{\pi}^2-k_0^2=\frac{4}{\ell_1}\;\frac{1}{\pi  n}\,\Big|\frac{\beta _2}{\beta _3}\Big|+\mathcal{O}(n^{-3}),
\end{equation}
as $n\rightarrow\infty$ which shows, in particular, that the band width behaves like $\mathcal{O}(n^{-1})$, that is, $\mathcal{O}(k^{-1})$ in terms of the corresponding momentum values. In the exceptional cases when $\beta_{3}$ vanishes, higher powers of $\delta$ in (\ref{eq,beta,Asy}) must be taken into account.

However, even in the generic case the gaps may exhibit a different behavior depending on the commensurability of the parameters involved:
\begin{itemize}
\item If $\frac{\ell_2}{\ell_1}$ and $\frac{\ell_3}{\ell_1}$ are both rational, $\beta_3$ as function of $n$ is periodic and, in particular, it has periodically distributed zeros. This may happen for all $n$ when all the lengths involved coincide, as is the case in Example~1 below, but if ratios are not that simple, there are $n$'s for which $\beta_3$ is nonzero. In view of the periodicity, those have the distance from zero bounded from below by a positive number, and therefore the $\mathcal{O}(n^{-1})$ asymptotics \eqref{bandwidth1} applies to them.
\item If, on the other hand, both the $\frac{\ell_2}{\ell_1}$ and $\frac{\ell_3}{\ell_1}$ are irrational, $\beta_3$ is nonzero for any $n$, but one can find sequences of gap indices along which $\beta_3$ tends to zero in correspondence with rational approximations of the given ratios and some band may be considerably wider than their neighboring ones.
\end{itemize}
In addition to that, of course, there are various mixed cases. To get a better idea, let us look at some examples:
\begin{example} \label{ex1,section}
{\rm The simplest situation refers to a particular case of the model considered in \cite{BET20} when $\ell_1=\ell_2=\ell_3=\pi$. Substituting these values into (\ref{beta,ell_1}), we get $\beta_1=\beta_2=\beta_3=\beta_6=\beta_7=\beta_8=0$ and
$$ 
\beta_4=  -\beta_5 = -4\pi ^3 (-1)^n \ell^2, \;\; \beta_9=\pi^3 (-1)^n \ell^4, \;\; \beta_{10} = -10\beta _{11}=\frac{20}{3} \pi^5 (-1)^n \ell^2.
$$ 
Substituting these values into \eqref{eq,beta,Asy}) we infer that
$$ 
\delta=\frac{2}{\ell \pi  n}\,\sqrt{1-\cos\theta}+\mathcal{O}(n^{-3}),
$$ 
referring to the energy $E_n(\theta)=k_\theta^{2}= n^2+\frac{4}{\pi\ell}\,\sqrt{1-\cos\theta}+\mathcal{O}(n^{-2})$ as a function of the quasimomentum, hence the band widths are
$$ 
\triangle E_n=k_{\pi }^2-k_0^2= \frac{4\sqrt{2}}{\pi\ell}  +\mathcal{O}(n^{-2}),
$$ 
being asymptotically constant as $n\rightarrow\infty$. As needed, the dimension of this quantity is that of inverted squared length, recall that $\ell_1=\pi$.
}
\end{example}

It may happen that the band widths have several nonzero limits as $n\rightarrow\infty$.
\begin{example} 
{\rm Let $\ell_1= \ell_3 = \frac{4}{11}\pi$ and $\ell_2=\frac{18}{11}\pi$ with $\ell=1$.  In this case, we have $\beta_{3}=0$ for all $n\in\mathbb{N}$ while $\beta_{4},\beta_{5},\beta_{10},\beta_{11}$ are all nonzero. The remaining coefficients depend on the parity: we have $\beta_{9}=0$ for odd $n$ and $\beta_{1}=\beta_{2}=\beta_{6}=\beta_{7}=\beta_{8}=0$ for the even ones. Consequently, \eqref{eq,beta,Asy}) implies two types of asymptotic behavior. For even $n$, solving
$$
\delta  \left(\frac{\beta _5 \cos \theta +\beta _4}{\pi ^2 n^2}\right)+\delta ^3 \left(\beta _9+\frac{\beta _{10}+\beta _{11} \cos \theta }{\pi ^2 n^2}\right) =0,
$$
we arrive at
$$
\triangle E_n=\frac{2}{\ell _1}\,\frac{\sqrt{\beta_5-\beta_4}-\sqrt{-\beta_4-\beta_5}}{\sqrt{\beta_9}}+\mathcal{O}(n^{-2}),
$$
which assumes two different values of the leading term, $\frac{11}{3\pi}\sqrt{11}\approx 3.8701$ and $\frac{11}{3\pi}(\sqrt{11}-2)\approx 1.537$ for $n=4m$ and $n=4m-2$, respectively. For odd $n$, on the other hand, solving
$$
\frac{\beta _2 \cos \theta +\beta _1}{\pi ^2 n^2}+\delta  \left(\frac{\beta _5 \cos \theta +\beta _4}{\pi ^2 n^2}\right)+\delta ^2 \left(\beta _6+\frac{\beta _{7}+\beta _{8} \cos \theta }{\pi ^2 n^2}\right) =0,
$$
we get two roots; using the fact that $\beta_1+\beta_2=0$ holds in this case, they both simplify to the same expression,
$$
\triangle E_n=\frac{2}{\ell_1}\;  \sqrt{\frac{\beta_2-\beta_1}{\beta_6}}   +\mathcal{O}(n^{-1}),
$$
in which that leading term value is $\frac{11}{\pi}\approx 3.501$. Hence all the band widths here are asymptotically constant assuming three different values as shown in Fig.~\ref{Figure05}.
\begin{figure}[h]
  \centering
\includegraphics[width=0.8\textwidth]{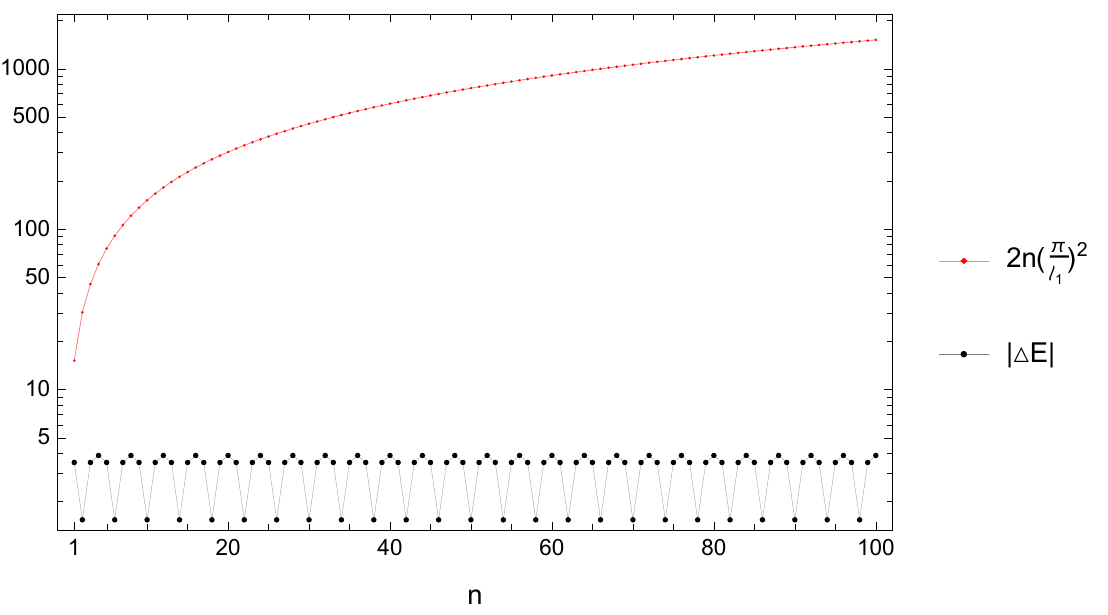}
  \caption{Band width leading term \emph{vs.} band index (referring to bands around the points $\frac{n\pi}{\ell_{1}}$) for $\ell_1= \ell_3 = \frac{4}{11}\pi$, $\ell_2=\frac{18}{11}\pi$, and $\ell=1$ in comparison with the gap width}
  \label{Figure05}
\end{figure}
The plot is rather simple but it is useful for comparison with the other examples given below. To make it more visible how the values are switching as the band index is changing, we draw here and in the following the lines joining points referring to the adjacent values of $n$. Recall also that the quantity plotted here is the leading term which represents a reliable characterization of the band width provided it is much smaller than the width of the adjacent \emph{gaps}. For comparison we plot here also the (leading term of the) latter; we see that the condition is satisfied, the better the large the band index is.
}
\end{example}

In other cases the asymptotically constant bands may be combined with shrinking ones.
\begin{example} 
{\rm Let now $\ell_1= \frac{\pi}{5}$, $\ell_2=13\frac{\pi}{7}$, $\ell_3=\frac{\pi}{7}$, and $\ell=1$, in which case $\beta_1= \beta_2=0$ for any $n\in\mathbb{N}$, while $\beta_{4},\beta_{5},\beta_{9},\beta_{10},\beta_{11}$ never vanish and $\beta_{3}=\beta_{6}=\beta_{7}=\beta_{8}=0$ holds for $n=7m$. In the latter case, we use \eqref{eq,beta,Asy}) and solve
$$
\delta\Big(\frac{\beta_5 \cos\theta +\beta_4}{\pi^2 n^2}\Big)+\delta^3 \left(\beta_9+\frac{\beta_{11} \cos\theta +\beta_{10}}{\pi^2 n^2}\right)=0\,;
$$
this yields bands of asymptotically constant widths appearing as spikes in Fig.~\ref{Figure06} below,
$$
\triangle E_n=\frac{2}{\ell_1}\,\frac{\sqrt{\beta_5-\beta_4}-\sqrt{-\beta_4-\beta_5}}{\sqrt{\beta_9}}+\mathcal{O}(n^{-2})
$$
with the leading term $\frac{28}{\pi}\sqrt{\frac{10}{13}}\approx 7.817$.
\begin{figure}[h]
  \centering
\includegraphics[width=0.8\textwidth]{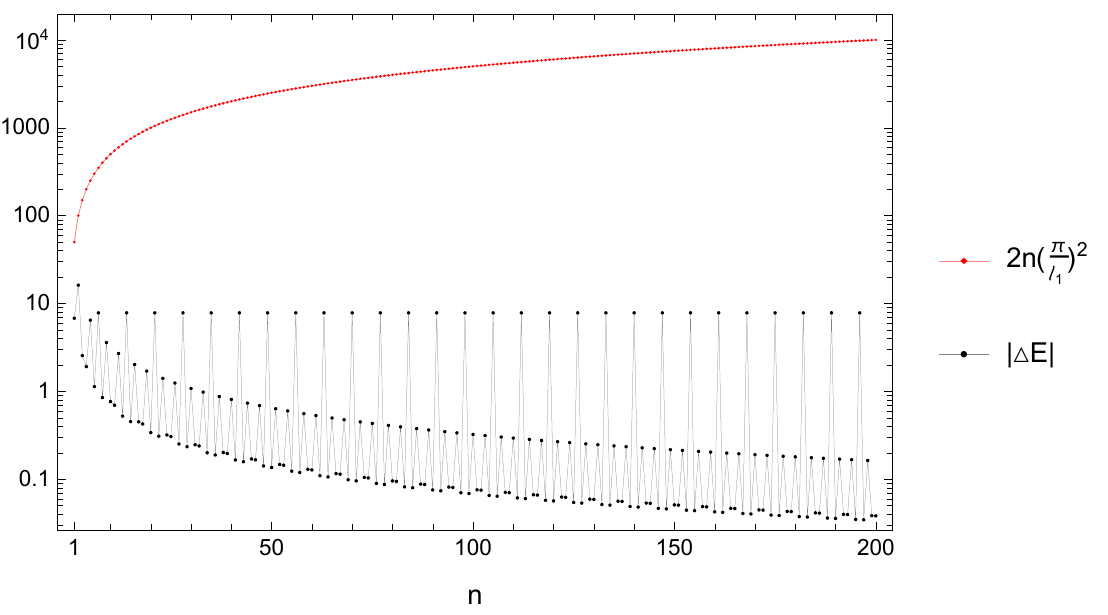}
  \caption{The leading term \emph{vs.} band index for $\ell_1= \frac{\pi}{5}$, $\ell_2=13\frac{\pi}{7}$, $\ell_3=\frac{\pi}{7}$, and $\ell=1$, again in comparison with the gaps width}
  \label{Figure06}
\end{figure}
On the other hand, if $n$ is not a multiple of seven, equation \eqref{eq,beta,Asy}) leads to
$$
\delta \left(\beta_3+\frac{\beta_5 \cos\theta +\beta_4}{\pi^2 n^2}\right)+\delta^2 \left(\beta_6+\frac{\beta_8 \cos\theta +\beta_7}{\pi^2 n^2}\right)=0,
$$
producing bands of decreasing widths,
\begin{equation} \label{bandwidth5638}
\triangle E_n=\frac{4\pi}{\ell_1}\,\frac{1}{n}\,\bigg|\frac{\beta_5 \beta_6-\beta_3 \beta_8}{\pi^2 \beta_6^2}\bigg|+\mathcal{O}(n^{-2}).
\end{equation}
Fig.~\ref{Figure06} shows that the coefficient of $n^{-1}$ in the leading term is periodic (mod 7) and that the gaps dominates as $n$ increases.
}
\end{example}

On the other hand, it may happen that \emph{`almost all' bands} are shrinking with respect to the band index which is the case when we make some of the edges incommensurate.
\begin{example} 
{\rm Choose $\ell_1=1$, $\ell_2=\ell_3=\pi$, and $\ell=1$, then the first three coefficients in \eqref{beta,ell_1} are
$$
\beta_{1}=2(-1)^{n+1} \sin 2n \pi^2,\quad \beta_{2}=4\sin n\pi^2, \quad \beta_{3}= (-1)^n \sin^2 n\pi ^2\,;
$$
since the second and the third do not vanish, the band width asymptotics is given by \eqref{bandwidth1}. Its dependence on the band index plotted in Fig.~\ref{Figure07} shows that even the `spikes' making the locally widest bands are decreasing.
\begin{figure}[h]
  \centering
\includegraphics[width=0.8\textwidth]{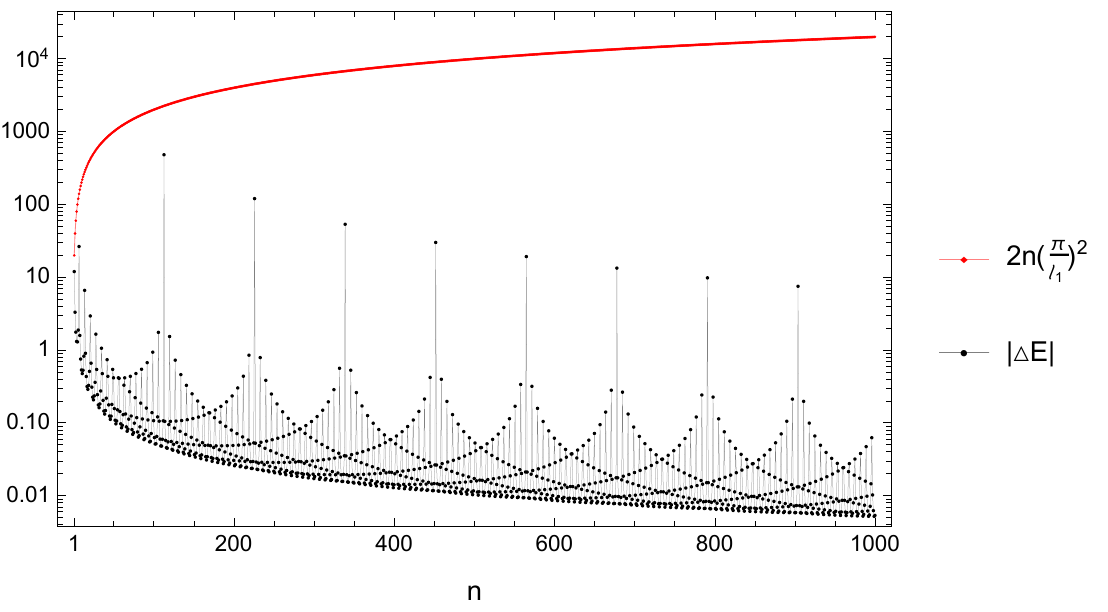}
  \caption{The leading term \emph{vs.} band index for $\ell_1=1$, $\ell_2=\ell_3=\pi$, and $\ell=1$}
  \label{Figure07}
\end{figure}
However, the used scale is too narrow to make a general conclusion; expanding the range of the index $n$ two or three orders of magnitude, one finds that there is a sequence of very rare bands the widths of which do not decrease. At such a large scale, we also see that the gaps eventually dominate even over such exceptional bands.
}
\end{example}

In the above example the band width plot still shows a lot of regularity due to the fact that $\ell_2=\ell_3$. Not surprisingly, this changes if make all the involved edges incommensurate.
\begin{example} \label{ex5,section}
{\rm Suppose that the two parts of the ring are in golden ratio, $\frac{\ell_2}{\ell_3}=\frac{\sqrt{5}+1}{2}$, and $\ell_1=\ell=1$. Then all the $\beta_i$'s are nonzero and the band widths are given by \eqref{bandwidth1}. On the other hand, changing the length of the connecting link to $\ell_1=\pi$ we get $\beta_1=\beta_2=0$ and formula \eqref{bandwidth5638} applies.
\begin{figure}[h]
  \centering
\includegraphics[width=0.75\textwidth]{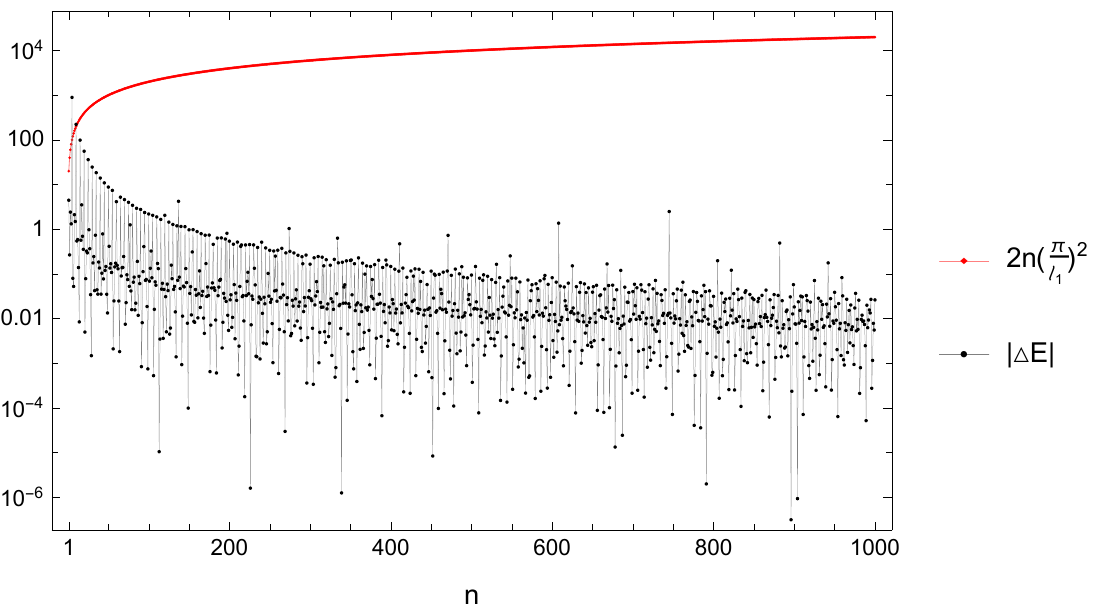}
  \caption{The leading term \emph{vs.} band index for $\ell_1=1$, $\frac{\ell_2}{\ell_3}=\frac{1+\sqrt{5}}{2}$, and $\ell=1$}
  \label{Figure08}
\end{figure}
\begin{figure}[h]
  \centering
\includegraphics[width=0.75\textwidth]{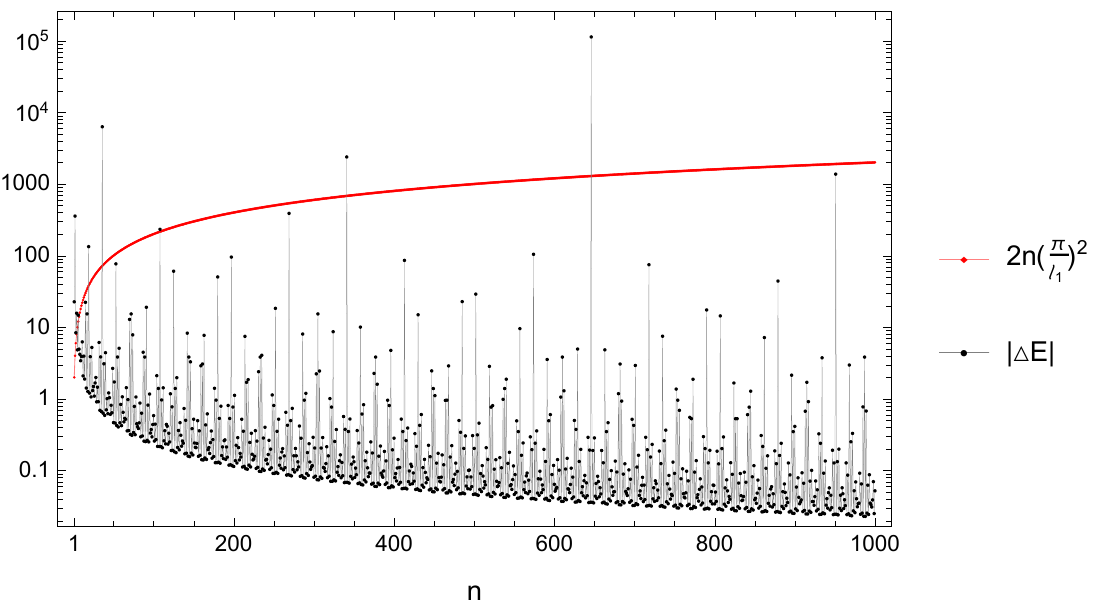}
  \caption{The same as above but with $\ell_1=\pi$}
  \label{Figure09}
\end{figure}
Then plots are irregular as seen from Figs.~\ref{Figure08} and \ref{Figure09}, and moreover, they depend on the choice of $\ell_1$; note that in the second case the latter is commensurate with the total length of the ring. At the same time, the influence of the parameter $\ell$ determining the length scale is less pronounced, changing it in the first case to $\ell=\pi$ we get a plot very similar to that of Fig.~\ref{Figure08}. The comparison with gap width mentioned in the previous examples applies again, however, the `extremely diophantine' character of the ratios involved means that we have to go very high values of the index, $n\gtrsim 10^6$, to have the gaps dominating over \emph{all} the bands.
}
\end{example}

\subsection{The other two band series}

To describe the bands around the points $\frac{n\pi}{\ell_j},\, j=2,3$, we proceed in a similar way so we limit ourselves to sketching the results without going into details. Setting again
$$ 
k_{n,j}=\frac{n\pi}{\ell_j}+\delta,\quad n\in\mathbb{N},\nonumber
$$ 
we write again the spectral condition in the form of the expansion \eqref{eq,beta,Asy} where the leading coefficients are now
\begin{align}\label{beta,ell_2}
&\beta_1=  2\ell^2 \ell_2^2 \left(\sin\frac{\pi n\ell_1}{\ell_2}+(-1)^{n+1} \sin\frac{\pi n \left(\ell_1+\ell_3\right)}{\ell_2}\right), \nonumber  \\[14pt]
&\beta_2= 2\ell^2 \ell_2^2 \sin\frac{\pi n\ell_3}{\ell_2}, \quad
\beta_3=  (-1)^n \ell ^4 \ell _2 \sin\frac{\pi n\ell_1}{\ell_2} \sin\frac{\pi n\ell_3}{\ell_2},
\end{align}
for $j=2$, and
\begin{align}\label{beta,ell_3}
&\beta_1= 2\ell^2 \ell_3^2 \left(\sin\frac{\pi n\ell_1}{\ell_3}+(-1)^{n+1} \sin\frac{\pi n\left(\ell_1+\ell_2\right)}{\ell_3}\right) , \nonumber  \\[14pt]
&\beta_2= 2 \ell^2 \ell_3^2 \sin\frac{\pi n\ell_2}{\ell_3}, \quad
\beta_3= (-1)^n \ell^4 \ell _3 \sin\frac{\pi n\ell_1}{\ell_3} \sin\frac{\pi n\ell_2}{\ell_3},
\end{align}
for $j=3$. In the generic case when $\beta_3$ stays away from zero the asymptotic behavior of the solution to \eqref{eq,beta,Asy} is given by \eqref{bandwidth1} again, hence replacing $\ell_1$ by $\ell_j$ and using the coefficients \eqref{beta,ell_2} and \eqref{beta,ell_3}, we arrive at
$$ 
\triangle E_n = \frac{8}{\pi n\ell^2}\, \Big|\sin\frac{\pi n\ell_1}{\ell_j}\Big|^{-1} +\mathcal{O}(n^{-3}),\quad j=2,3,
$$ 
as $n\rightarrow\infty$. In the other cases we have to resort to higher powers of $\delta$ in (\ref{eq,beta,Asy}) similarly as we did that in the previous section.

\subsection{The band width asymptotic behavior: an overall view}

The entire positive spectrum naturally combines all the bands mentioned above. The band structure depends on the ratios of the lengths involved. To get a better idea, instead of combining the results of the previous two sections, we compute the spectrum directly from \eqref{band,abc,gen,pos} in two situations corresponding to Examples~\ref{ex1,section} and \ref{ex5,section}; the results are plotted in Figs.~\ref{Figure10} and \ref{Figure11}, respectively.
\begin{figure}[h]
  \centering
\includegraphics[width=0.9\textwidth]{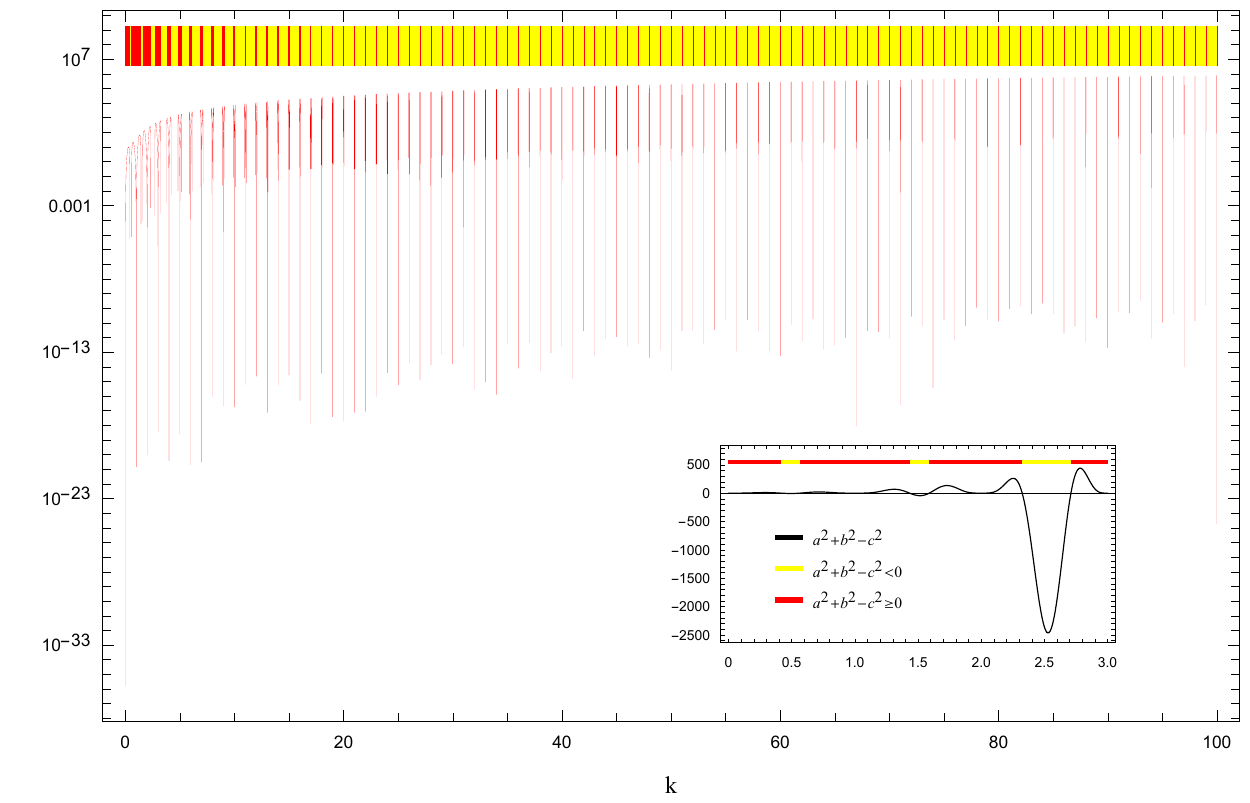}
  \caption{The full band spectrum in the situation of Example~\ref{ex1,section} with $\ell=1$}
  \label{Figure10}
\end{figure}
We see the difference between the two situations: in the first case when the edge lengths are the same the band pattern is periodic in the momentum scale, while irrational relations produce more bands (roughly three times as many in view of \eqref{asymptcond}) distributed in an irregular way.
\begin{figure}[h]
  \centering
\includegraphics[width=0.9\textwidth]{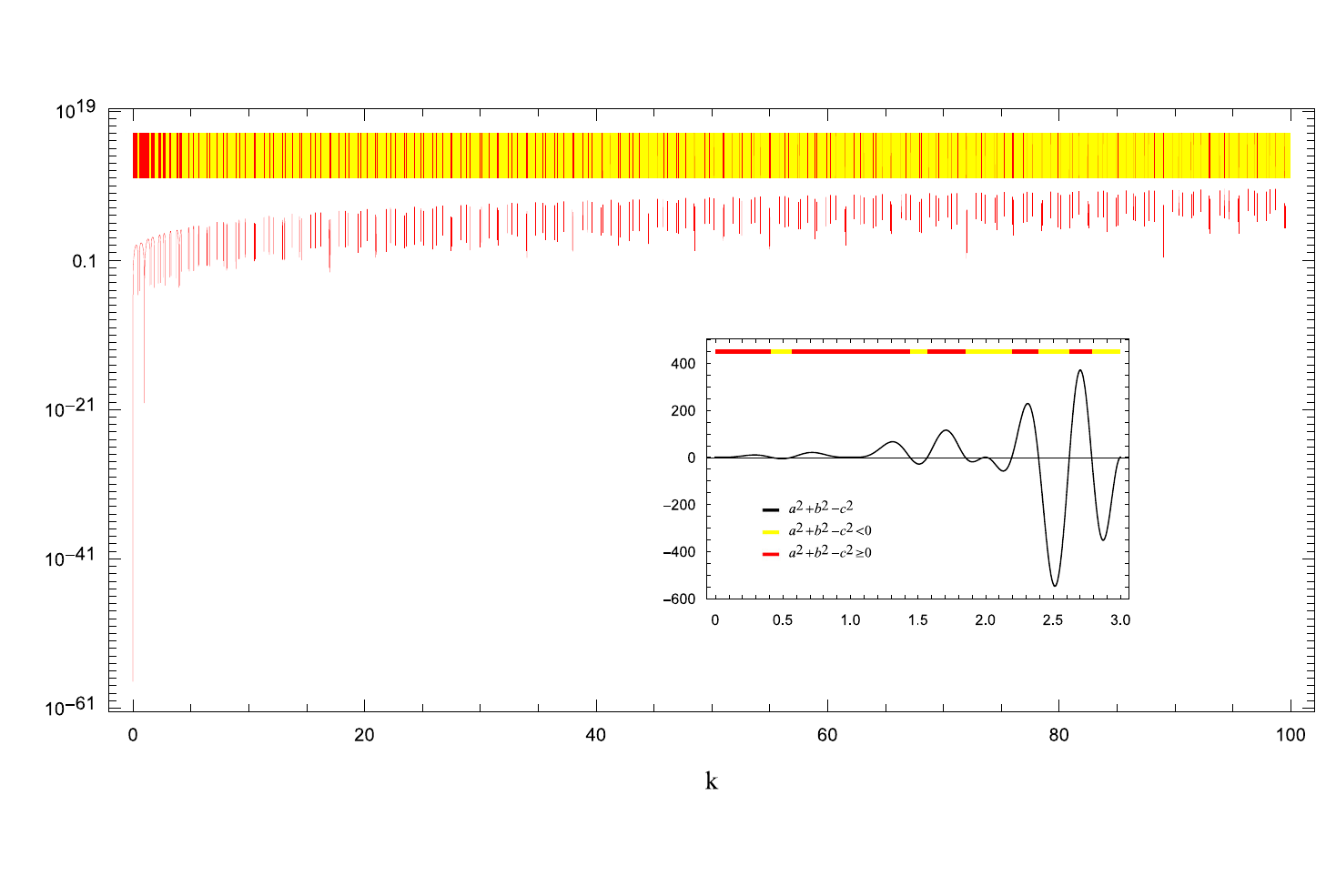}
  \caption{The full band spectrum in the situation of Example~\ref{ex5,section} with $\ell_1=\pi$}
  \label{Figure11}
\end{figure}
To get an understandable picture we choose a more narrow momentum range than in Figs.~\ref{Figure07}--\ref{Figure09}. On the other hand, the obtained band patterns are exact, not only asymptotic. The curves below them show positive values of $a^2+b^2-c^2$ as a function of $k$. Since their shape is still not well seen at the chosen scale, we show in the insets of Figs.~\ref{Figure10} and \ref{Figure11} this function -- its values now in the linear, not logarithmic scale -- over a shorter interval corresponding to the lowest four or five bands.

Despite their different structure, the band spectra in all cases have a common property. As we have noted in the opening of Sec.~\ref{ss:asympt}, the bands can occur (at the momentum scale) only in the vicinity of the points \eqref{Sp,points}. As a union of three periodic sequences, the density of this set is uniform in the sense that the number of its points in an interval $(k,k')$ tends to $\big(2+\frac{\ell_1}{\pi}\big)|k'-k|$ as $|k'-k|\to\infty$. At the same time, it follows from \eqref{bandwidth1} that the band widths are $\mathcal{O}(k^{-1})$ as $k\to\infty$. Translated to the energy scale, $k^2$, this means that the density of points $k^2_{j,n}$ decreases as $\mathcal{O}(k^{-1})$ while the bands are at most constant. In terms of the \emph{probability of belonging the (positive) spectrum} introduced by Band and Berkolaiko \cite{BB13} it means that
\begin{equation}\label{prob,l>0}
P_{\sigma}(H):=\lim_{K\to\infty} \frac{1}{K}\left|\sigma(H)\cap[0,K]\right| = 0,
\end{equation}
where, of course, zero can be replaced here by any fixed positive number.

\subsection{The negative spectrum}

The condition determining the spectrum in the negative half of the real axis can be obtained from \eqref{Pos,SC,ell_1,3} by replacing $k$ with $i\kappa$; this gives
\begin{align}\label{Neg,SC,ell_1,3}
& 16\sinh\pi\kappa \big(\cos\theta (\kappa^2\ell^2-1) \cosh\kappa(\pi-\ell_3) +2\kappa\ell \sin\theta \sinh\kappa(\pi-\ell_3)\big)\nonumber\\
& -(\kappa^2\ell^2+1)^2 \big(\sinh\kappa(2\pi-\ell_1) +2\sinh\kappa\ell _1 \cosh 2\kappa(\pi-\ell_3)\big)\nonumber\\
& +8(\kappa^2\ell^2-1) \sinh\kappa\ell_1 +(\kappa^2\ell^2-3)^2 \sinh\kappa(\ell_1+2\pi)=0.
\end{align}
Mimicking the argument that led to \eqref{band,abc,gen,pos}, we arrive at the band condition
\begin{equation}\label{band,Neg,Gen}
128\,\sinh^2\pi\kappa\: \big(\kappa^4\ell^4 -6\kappa^2\ell^2 +(\kappa^2\ell^2+1)^2 \cosh 2\kappa(\pi-\ell_3)+1\big) -(\rho-\tau)^2 \ge 0,
\end{equation}
where
\begin{align*}
& \rho:= (\kappa^2\ell^2+1)^2 \big(\sinh\kappa(2\pi-\ell_1) +2\sinh\kappa\ell_1 \cosh 2\kappa(\pi-\ell_3)\big), \\
& \tau:= (\kappa^2\ell^2-3)^2 \sinh\kappa(\ell_1+2\pi) +8(\kappa^2\ell^2-1) \sinh\kappa\ell_1.
\end{align*}
Note that in this case the analogue $a^{2}+b^{2}$ in (\ref{eq,sin,al,th}) is nonzero for any $\kappa>0$, and consequently, there are no flat bands in the negative part of the spectrum.

Extending slightly the argument from \cite{BET20} we see that in the symmetric case, $\ell_2=\ell_3=\pi$, the interval $(-\ell^{-2},0)$ does not belong to the spectrum. Indeed, condition \eqref{Neg,SC,ell_1,3} then becomes
$$ 
\cos\theta= \cosh\kappa\pi \cosh\kappa\ell_1 +\sinh\kappa\pi \sinh\kappa\ell_1 \:\frac{\kappa^4\ell^4-2\kappa^2\ell^2+5}{4-4\kappa^2\ell^2},
$$ 
and since the last fraction is positive for $\kappa<\ell^{-1}$, the right-hand side is then larger than one. This no longer true generally, but we can still claim that
\begin{itemize}
\item the negative spectrum remains separated from zero.
\end{itemize}
To see that, let us inspect the behavior of the left-hand side of \eqref{band,Neg,Gen} as $\kappa\to 0+$ taking its Taylor expansion to the fourth order,
$$ 
-128\pi(\ell_1+2\pi) \big((2\pi-\ell_3)\ell_3 +2\pi\ell_1\big)\kappa ^4 +\mathcal{O}(\kappa ^6),
$$ 
in which the leading term is negative for small $\kappa$ and any fixed choice of the paramaters, uniformly across the Brillouin zone.

The next question concerns the number of negative bands. In the particular case considered in \cite{BET20} we found an answer through a direct investigation of the spectral condition. One can approach, however, the problem from a more general point of view. In order to state the result, let us return to the vertex condition \eqref{genbc}. If a vertex is of degree $n$, the matrix $U$ characterizing the coupling has $n$ eigenvalues, with their multiplicity taken into account. We divide them into three groups, $n^{(0)}$ real eigenvalues (necessarily equal to $\pm 1$ corresponding to the Neumann and Dirichlet part of the coupling in the sense of Theorem 1.4.4 in \cite{BK13}, respectively), and $n^{(\pm)}$ non-real `Robin' eigenvalues situated in the upper and lower complex plane, respectively.
\begin{theorem} \label{thm:negspect}
Consider a periodic quantum graph and assume that its elementary cell contains $N$ vertices with the couplings described by unitary matrices $U_i,\: i=1\,\dots,N$, then the negative spectrum of the corresponding Hamiltonian consists of at most $\sum_{i=1}^N n_i^{(+)}$ bands.
\end{theorem}
\begin{proof}
Since the dependence of the fiber operators $H(\theta)$ in \eqref{Floquet1} on the quasimomentum is continuous, the number of negative bands coincides with the number of dispersion curves in the negative part of the axis, i.e. with the maximum number of negative eigenvalues that $H(\theta)$ may have. It is easy to see that a star graph Hamiltonian with the vertex of degree $n$ characterized by a matrix $U$ has exactly $n^{(+)}$ negative eigenvalues, their multiplicity taken into account; if $n^{(+)}=0$, the corresponding operator is positive. What is important is that such a positivity connected with a vertex coupling is a local property, as seen from the corresponding quadratic form given in Theorem 1.4.11 of \cite{BK13}.

Consider now the fiber $H(\theta)$ corresponding to our periodic graph with a fixed $\theta\in[-\pi,\pi)$. The total number of the eigenvalues of $U_i$ situated in the open lower halfplane is $N_+:= \sum_{i=1}^N n_i^{(+)}$. Let $\tilde H = \int_{-\pi}^\pi \tilde H(\theta)\, \mathrm{d}\theta$ refer to a quantum graph in which each vertex matrix $U_i$ is replaced by $\tilde U_i$ such that its eigenvalues with non-negative imaginary parts are preserved and those with negative ones are replaced with eigenvalues which are either real or situated in the upper halfplane. It follows from this construction that the self-adjoint operators $H(\theta)$ and $\tilde H(\theta)$ have a common symmetric restriction with the deficiency indices $(N_+,N_+)$. From the general theory of self-adjoint extensions it follows that the spectrum of $H(\theta)$ in any spectral gap of $\tilde H(\theta)$ consists of at most of isolated eigenvalues of the total multiplicity $N_+$, cf. Corollary~1 to Theorem~8.19 in \cite{We80}. However, $\tilde H(\theta)$ is positive by construction, hence $H(\theta)$ cannot have more than $N_+$ negative eigenvalues, their multiplicity take into account.
\end{proof}

This general result confirms what we know for particular cases; recall the examples of square and honeycomb lattices \cite{ET18} having one and two negative bands, respectively. In our present situation we have two vertices of degree three in the elementary cell; the corresponding matrix $U$ in each of them has eigenvalues $1$ and $\e^{\pm 2\pi i/3}$ so that $N_-=2$. It should be stressed that it may happen that the two bands merge into one. In the symmetric case considered in \cite{BET20} we has shown that this happens for $\ell=1$ and $\ell_1=\pi$. The same is true for any $\ell$; to see that it is sufficient to note that for $\ell_1=\ell_3=\pi$ the left-hand side of \eqref{band,Neg,Gen}  becomes
$$
-8 \sinh^4 \pi\kappa (\kappa^2\ell^2-3)^2 \big(-\kappa^4\ell^4 -10\kappa^2\ell^2 +\cosh 2\pi\kappa (\kappa^2\ell^2-3)^2+7\big)
$$
vanishing at $\kappa= \sqrt{3}\ell^{-1}$ and being positive in the vicinity of this point. For other values of the parameters the two bands may merge at a different value of $\kappa$ but an overall picture remains similar as illustrated on Fig.~\ref{Figure12}.
\begin{figure}[h]
\centering
\includegraphics[scale=.75]{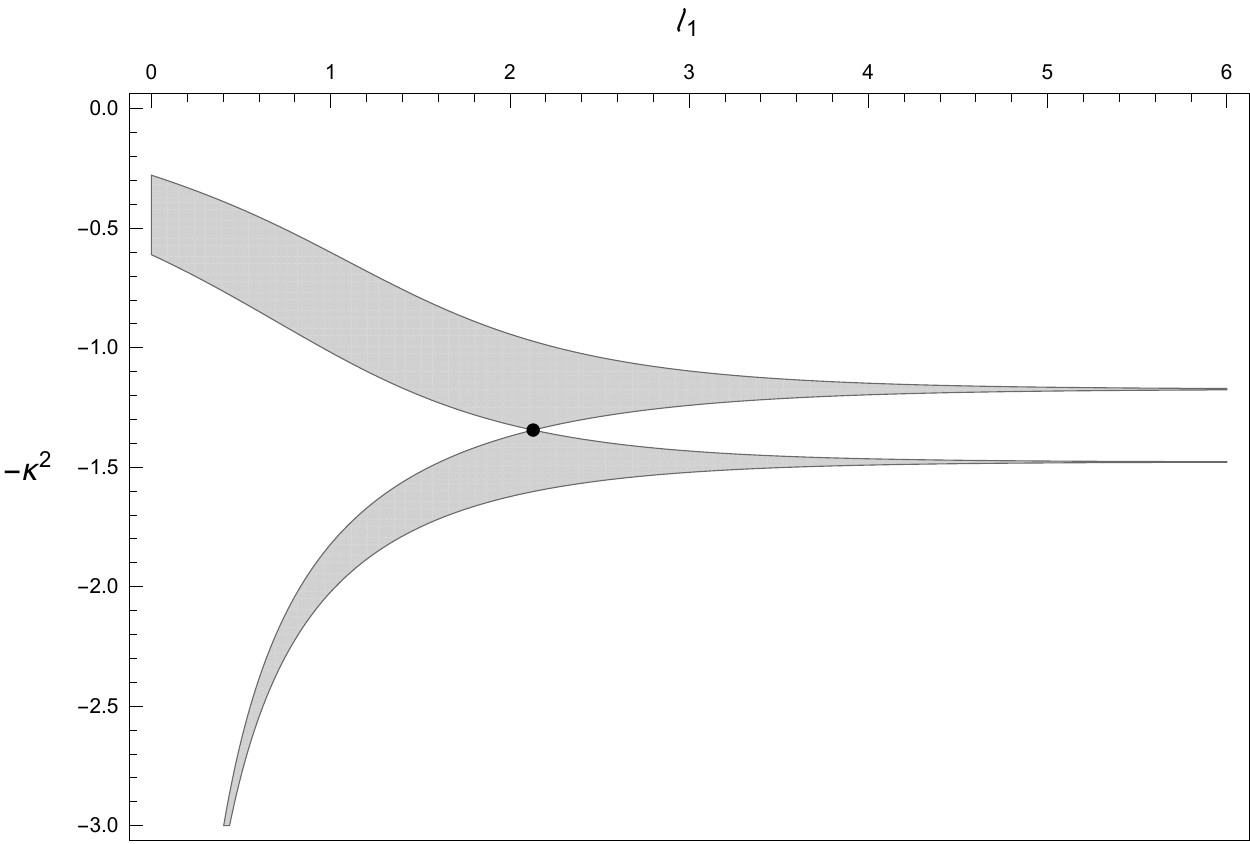}
\caption{The negative spectrum for $\ell=\frac{3}{2}$ and $\ell_3=\frac{2}{3}\pi$ as a function of $\ell_1>0$, the dot indicates the band crossing point.}
\label{Figure12}
\end{figure}
The band crossing point is given by a relation analogous to \eqref{cross,con}, however, in the general case we are unable to find its solution in a closed form.

\subsection{Long connecting links}

While the ring circumference is fixed by assumption, one may ask what happens if the value of $\ell_{1}$ is large. It is clear that the positive spectrum will become denser being determined primarily by bands around the points $\frac{\pi n}{\ell_1}$. Nevertheless, the chain will be generically non-conducting since by \eqref{prob,l>0} the probability that an energy lies in the positive spectrum is zero for any $\ell_1$.

The negative spectrum, on the other hand, consists of two bands (which may merge at some point) for any $\ell_1$; one can ask how they behave when $\ell_1\to\infty$. Since we have $\sinh\kappa\ell_1 \approx \cosh\kappa\ell_1 \approx \frac12\,\e^{\kappa\ell_1}$ for a fixed $\kappa>0$ in this limit, we can write the spectral condition \eqref{Neg,SC,ell_1,3} as
$$ 
f(\ell_3,\ell;\kappa)\,\e^{\kappa\ell_1} + g(\ell_3,\ell;\kappa,\theta) + \OO(\e^{-\kappa\ell_1})=0
$$ 
with
$$ 
g(\ell_3,\ell;\kappa,\theta):= 32\,\sinh\pi\kappa \big(\cos\theta (\kappa^2\ell^2-1) \cosh\kappa(\pi-\ell_3) +2\kappa\ell\sin\theta \sinh\kappa(\pi -\ell_3)\big).
$$ 
This shows that in the limit the bands shrink, exponentially fast, to the points determined by the condition $f(\ell_3,\ell;\kappa)=0$, or explicitly
\begin{equation}\label{NegLargeL1}
e^{2\pi\kappa}\,(\kappa^2\ell^2-3)^2 +8(\kappa^2\ell^2-1)+(\kappa^2\ell^2+1)^2 \big(e^{-2\pi\kappa} -2\cosh2\kappa(\pi-\ell_3)\big)=0.
\end{equation}
Those are situated at both sides of $\kappa =\frac{\sqrt{3}}{\ell}$; one checks easily that
$$ 
f(\ell_3,\ell;\frac{\sqrt{3}}{\ell}) = 16 \big(1+\e^{-\frac{2\sqrt{3}\pi}{\ell}} -2\cosh\frac{2\sqrt{3} (\pi-\ell_3)}{\ell} \big)<0
$$ 
in view of the inequalities $\cosh x\geq1$ and $\e^{-\frac{1}{x}}\leq1$ for $x>0$, and on the other hand, we have $f(\ell_3,\ell;\kappa) =16\kappa\pi +\big(20\pi^2 -4(\pi-\ell_3)^2\big)\kappa^2 +\mathcal{O}(\kappa^{3})>0$ for small $\kappa$ and $\lim_{\kappa\to\infty} f(\ell_3,\ell;\kappa)=+\infty$; according to the general result obtained above the equation \eqref{NegLargeL1} cannot have more than two roots.

In the limit $\ell_1\to\infty$ the bands shrink to points. In the symmetric case, $\ell_3=\pi$, assuming that the scale parameter $\ell$ is not too large we can substitute $\kappa^{2}=\frac{3}{\ell^{2}}+\varepsilon$ into \eqref{NegLargeL1} and obtain
$$ 
\varepsilon \approx \pm 4 \ell^{-2} \big(\e^{\frac{\sqrt{3} \pi }{\ell }}\mp1 \big)^{-1},
$$ 
in particular, $\varepsilon\approx \pm 0.0173$ for $\ell=1$ in accordance with the result of \cite{BET20}. For $\ell_3\ne\pi$ we move away from this approximate symmetry, in the limit $\ell_3\to 0$ the upper limit point approaches $-\kappa^{2}=-1/\ell^{2}$ as the condition \eqref{NegLargeL1} reduces to $-8(\e^{2\pi\kappa}-1) (\kappa^2\ell^2-1)=0$ while the lower one escapes to $-\infty$.

\section{Vertices of degree four: the case of $\ell_{1}=0$}

Let us now pass to the situation where the vertices are of degree four. This may happen if one of the edge lengths, $\ell_1$ or $\ell_3$ shrinks to zero (not both of them, of course). Consider first the former situation. To find the spectral condition, we can take the natural Ans\"atze and match them as we did when deriving the condition \eqref{Pos,SC,ell_1,3}, or equivalently, to take the limit $\ell_{1}\rightarrow0$ in the latter; this yields
\begin{equation}\label{sp,con,ell_1=0}
\sin\pi k \big((k^2\ell^2\!+\!1) \big(\cos\theta \cos k(\pi\!-\!\ell_3)-\cos\pi k\big) +2 k\ell \sin\theta \sin k(\pi\!-\!\ell_3)\big)=0.
\end{equation}

\subsection{The positive spectrum}

The easiest thing to conclude from \eqref{sp,con,ell_1=0} is that $k^{2}$ belongs to the spectrum for any natural number $k=n\in\mathbb{N}$; it may or may not be embedded as we will see below. As for the continuous component, in the symmetric case, $\ell_3=\pi$, the condition (\ref{sp,con,ell_1=0}) reduces to $\cos\theta =\cos k\pi$, so that the spectrum covers the entire interval $[0,\infty)$. In the absence of this symmetry, the spectrum has a band-gap structure determined by vanishing of the bracket in \eqref{sp,con,ell_1=0}. Rewriting this condition in the form \eqref{eq,abc,gen}, the $a$, $b$, and $c$ are as follows
\begin{align*}
a & = (k^2\ell^2+1) \cos k(\pi-\ell_3), \\
b & = 2k\ell \sin k(\pi-\ell_3), \\
c & = (k^2\ell^2+1) \cos k\pi,
\end{align*}
where $a^{2}+b^{2}$ is obviously nonzero. Consequently, it follows from \eqref{band,abc,gen,pos} that $k^{2}$ belongs to a spectral band if and only if
\begin{equation}\label{gap,pos,ell_1=0}
4k^2\ell^2 -(k^2\ell^2+1)^2 \cos2\pi k +(k^2\ell^2-1)^2 \cos2k(\pi-\ell_3) \ge 0.
\end{equation}
In the high energy regime, in particular, the band condition reduces to the inequality $\cos 2k(\pi-\ell_3) \geq \cos 2k\pi$ which can be equivalently rewritten as
\begin{equation} \label{tightband1}
\sin k\ell_2 \sin k\ell_3 \ge 0
\end{equation}
keeping in mind that $\ell_2+\ell_3=2\pi$, and gaps refer to the values $k^2$ for which these expressions become negative, both modulo an $\OO(k^{-2})$ relative error.

It is clear from \eqref{tightband1} that the spectrum has an infinite number of open gaps whenever the chain is asymmetric, $\ell_3\ne\pi$. A detailed shape of the spectrum depends on the ratio $\frac{\ell_2}{\ell_3}$. It is clear that the band pattern is periodic at the momentum scale if this number is rational and irregular in the opposite case. It has, however, a universal property:
\begin{itemize}
\item The probability \eqref{prob,l>0} of belonging to the (positive) spectrum is
\begin{equation}\label{tightprob,l>0}
P_{\sigma}(H)= \left\{ \begin{array}{lcl} 1 & \quad\dots\quad & \ell_3=\pi \\[.3em] \textstyle{\frac12} & \quad\dots\quad & \ell_3\ne\pi \end{array} \right.
\end{equation}
\end{itemize}
Indeed, the probability that each of the factors on the left-hand side in \eqref{tightband1} is for a randomly chosen value of $k$ positive (or negative) is $\frac12$. In the symmetric case these effects are correlated, in the asymmetric they are not, so the probability is $(\frac12)^2 + (\frac12)^2 = \frac12$, and the band and gaps sizes thus grow at the energy scale in the average at the same rate. This means that due to its geometry an asymmetric chain is `less conductive'. Note the system is so simple that, in contrast to \cite{BB13}, the ergodicity is not needed; the result holds irrespective of whether $\ell_3/\pi$ is irrational.

On the other hand, a comparison of \eqref{prob,l>0} and \eqref{tightprob,l>0} shows again the effect observed in the particular situation in \cite{BET20}: for a chain of `loosely coupled' rings we have $P_{\sigma}(H)=0$ whatever the length $\ell_1$ of the connecting link may be, while in the limit $\ell_1\to 0$ the probability is positive (irrespective of the junction position). In other words, the spectrum converges in the vanishing edge limit in accordance with the result of \cite{BLS19} but the convergence is rather nonuniform.

\subsection{Gap closing and embedded eigenvalues}
\label{cross,sec,ell_1=0}

The existence of infinitely many gaps in the asymmetric case does not mean, of course, that some of them may not close for particular values of the parameters. Specifically, we have the following claims:

\begin{itemize}
  \item If the ring arcs are commensurate, $\frac{\ell_{2}}{\ell_{3}}\in \mathbb{Q}$, there are infinitely many points where the gaps close coinciding with some of the flat bands mentioned above; the latter can in this way be embedded in the continuous spectrum. These crossings occur at $(\ell_3,k)=\big(\frac{m\pi}{n}, nj \big)$ with $m,n,j\in\mathbb{N}$. To show that, we denote the expression in the large bracket at the left-hand side of (\ref{sp,con,ell_1=0}) by $\mathcal{H}(k,\ell_3,\theta)$. Then, in analogy with \eqref{cross,con}, crossings occur if
\begin{equation}\label{cross,ell1=0}
\frac{\partial \mathcal{H}(k,\ell_3,\theta)}{\partial \ell_3}=\frac{\partial \mathcal{H}(k,\ell_3,\theta)}{\partial \theta}=\frac{\partial \mathcal{H}(k,\ell_3,\theta)}{\partial k}=0.
\end{equation}
Evaluating the derivatives at $k=jn$ and $\ell_3= \frac{m\pi}{n}$, we get
\begin{align*}
& \frac{\partial \mathcal{H}(k,\ell_3,\theta)}{\partial\ell_3}\biggr|_{\substack{k=nj\\ \ell_3=m\pi/n}} = -2\sin\theta\; j^2n^2 \ell\, \cos \pi j (n-m), \\
& \frac{\partial \mathcal{H}(k,\ell_3,\theta)}{\partial \theta}\biggr|_{\substack{k=nj\\ \ell_3=m\pi/n}} = -\sin\theta (j^2n^2 \ell^2+1) \cos \pi  j(n-m), \\
& \frac{\partial \mathcal{H}(k,\ell_3,\theta)}{\partial k}\biggr|_{\substack{k=nj\\ \ell_3=m\pi/n}} =  2j\ell \cos\pi jn\, \big(\cos\pi jm\, \big(\pi(n-m)\sin\theta +n\ell\cos\theta\big) -n\ell \big).
\end{align*}
The first two expressions vanish for $\theta=0,\pm\pi$; substituting these values into in the third one we arrive at the equations
$$ 
-4jn\ell^2\, \sin^2\left(\frac{\pi jm}{2}\right) \cos\pi jn=0
$$ 
and
$$ 
-2jn\ell^2 \,(\cos\pi jm +1\big) \cos\pi jn=0,
$$ 
respectively; the former is satisfied for even values of $jm$, the latter for odd ones. This proves the claim and shows at the same time that the crossings occur at the edges of the Brillouin zone or in its center, depending on the parity of $jm$.
\item There can be no gap at the momentum value $k=\ell^{-1}$. To see that, we have to solve the equations \eqref{cross,ell1=0} at $(\ell^{-1},\ell_3,\theta)$. The derivatives are easily calculated, vanishing of the first two requires $\sin\frac{\theta\ell+\ell_3-\pi} {\ell }=0$ and third one gives
$$ 
\left(\pi -\ell _3\right) \sin \frac{\theta  \ell +\ell _3-\pi }{\ell } +\ell \cos  \frac{\theta  \ell +\ell _3-\pi }{\ell } +\pi  \sin  \frac{\pi }{\ell } -\ell  \cos  \frac{\pi }{\ell } =0.
$$ 
    The first term at the left-hand side vanishes and we get the equation
$$ 
\pm\ell\left(1\mp\cos\frac{\pi}{\ell}\right) +\pi\sin\frac{\pi}{\ell} = 0
$$ 
    which is satisfied for $\ell^{-1}\in\mathbb{N}$, the upper sign for even $\ell^{-1}$ and the lower sign for the odd one. Needless to say, this is the only missing gap in case that $\ell_2$ and $\ell_3$ are incommensurate; the band-and-gap pattern in dependence on $\ell_3$ is illustrated in Fig.~\ref{Figure13}.
\end{itemize}
\begin{figure}[!htb]
\centering
\includegraphics[scale=.7]{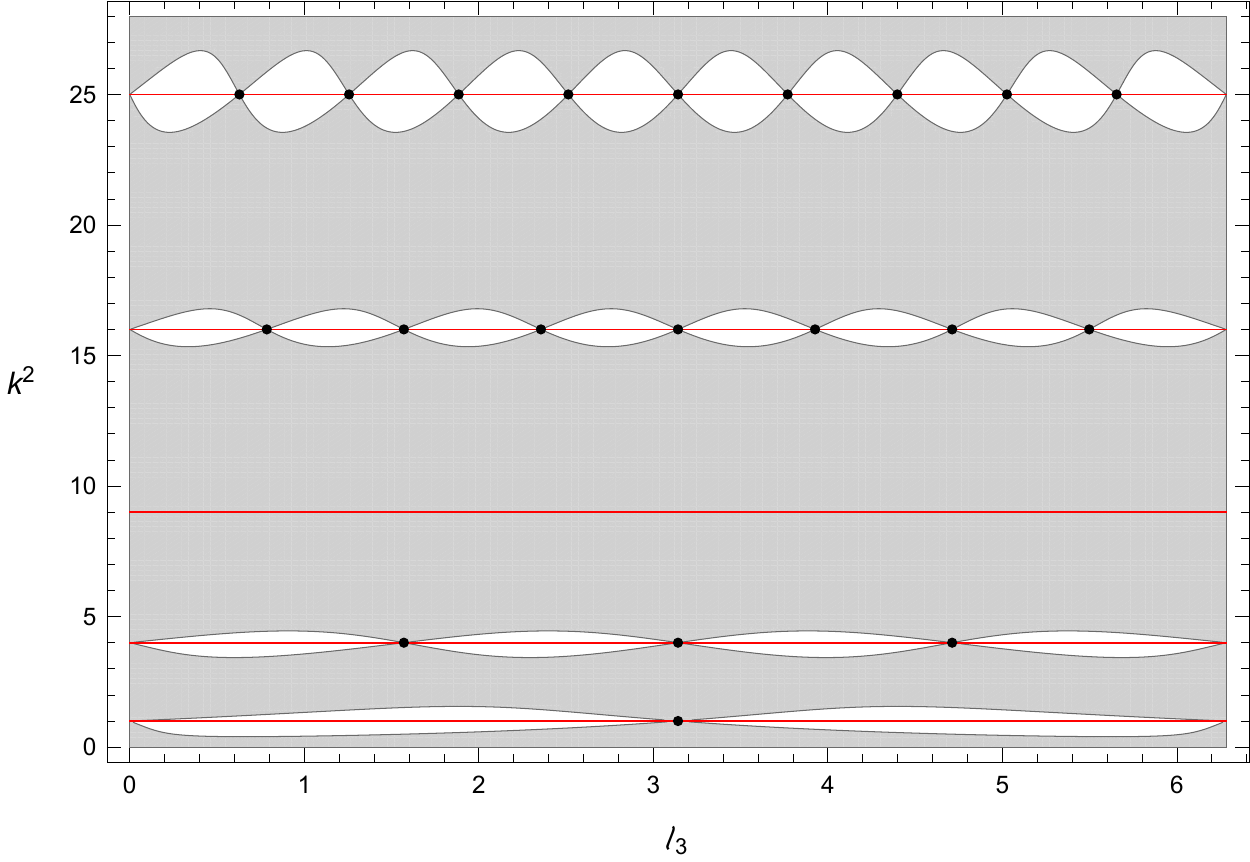}
\caption{The positive spectrum for $\ell=\frac13$.}
\label{Figure13}
\end{figure}

\subsection{The negative spectrum}

To find the negative spectrum one has to substitute $k=i\kappa$ into \eqref{sp,con,ell_1=0} and \eqref{gap,pos,ell_1=0}, in particular, it corresponds to the values $\kappa$ for which
$$ 
H(\kappa):= -(\kappa^2\ell^2-1)^2 \cosh 2\kappa\pi +(\kappa^2\ell^2+1)^2 \cosh 2\kappa (\pi-\ell_3) -4\kappa^2\ell^2
$$ 
is non-negative. The elementary cell now contains a single vertex of degree four and the eigenvalues of the corresponding matrix $U$ are $\pm 1$ and $\pm i$, hence by Theorem~\ref{thm:negspect} there is one negative band. It always contains the point $-\ell^{-2}$ because the inequality $(\kappa^2\ell^2+1)^2 \geq 4\kappa^2\ell^2$ implies $H(\ell^{-1})\ge 0$. In the symmetric case, $\ell_3=\pi$, we have $H(\kappa)= -2(\kappa^2\ell^2-1)^2 \sinh^2\kappa\pi$, so the negative band is flat as we already know from \cite{BET20}, otherwise $H(\ell^{-1})>0$ and the negative spectrum is absolutely continuous.

Furthermore, the band is symmetric with respect to the exchange of $\ell_3$ to $2\pi-\ell_3$ because $\cosh$ is an even function. Also, the negative spectrum remains separated from zero, since $H(\kappa)=-2\ell_3(2\pi-\ell_3)\kappa^{2}+\mathcal{O}(\kappa^{4})$ is negative for small values of $\kappa$. Large values of $\kappa$ also cannot give rise to spectral points because $\lim_{\kappa\to\infty} \,H(\kappa)=-\infty$ holds for any $\ell_3\in(0,2\pi)$.
\begin{figure}[!htb]
\centering
\includegraphics[scale=.7]{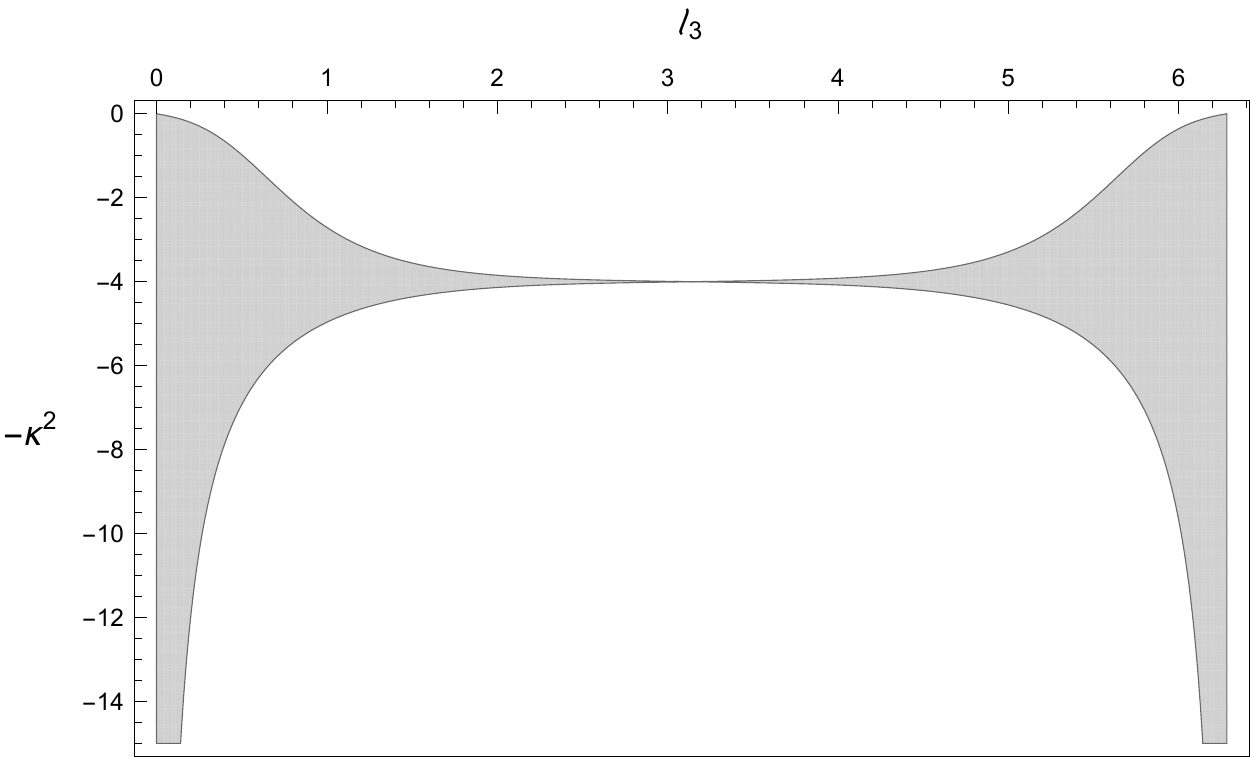}
\caption{The band dependence on $\ell_3$ for $\ell=\frac{1}{2}$.}
\label{Figure14}
\end{figure}
At the same time the band widens as illustrated in Fig.~\ref{Figure14} and its lower edges escapes to $-\infty$. This is no paradox, however, as the lengths cannot vanish simultaneously because is such a case the present model loses meaning.

\section{Vertices of degree four: the case of $\ell_{3}=0$}

\subsection{The positive spectrum}

To get the spectral condition, it is again sufficient to perform the limit $\ell_3\to 0$ in equation \eqref{Pos,SC,ell_1,3}, which yields
\begin{equation}\label{Pos,SC,ell_3=0}
\sin\pi k\, \big((k^2\ell^2+1) \big(\cos\theta \cos k\pi -\cos k(\ell_1+\pi)) +2k\ell \sin\theta \sin k\pi \big)=0.
\end{equation}
It is obvious that $k^2$ for any $k\in\N$ belongs to the spectrum independently of $\theta$. These flat bands are always embedded in the continuous spectrum; note that the large bracket at the left-hand side of \eqref{Pos,SC,ell_3=0} can be annulated by choosing $\cos\theta =\cos n\ell_1$. To determine the band structure beyond the flat case, we restate the requirement in the form \eqref{eq,abc,gen} with the coefficients
\begin{align*}
a & = (k^2\ell^2+1) \cos k\pi, \\
b & = 2k\ell \sin k\pi, \\
c & = (k^2\ell^2+1) \cos k(\ell_1+\pi),
\end{align*}
satisfying $a^{2}+b^{2}\neq0$. After some simplifications, we conclude that $k^{2}$ belongs to a band if and only if
\begin{equation}\label{Gap,Pos,ell_3=0}
4k^2\ell^2 +(k^2\ell^2-1)^2 \cos 2\pi k -(k^2\ell^2+1)^2 \cos 2k(\ell_1+\pi) \geq 0.
\end{equation}
The band pattern is illustrated in Fig.~\ref{Figure15}.
\begin{figure}[!htb]
\centering
\includegraphics[scale=.7]{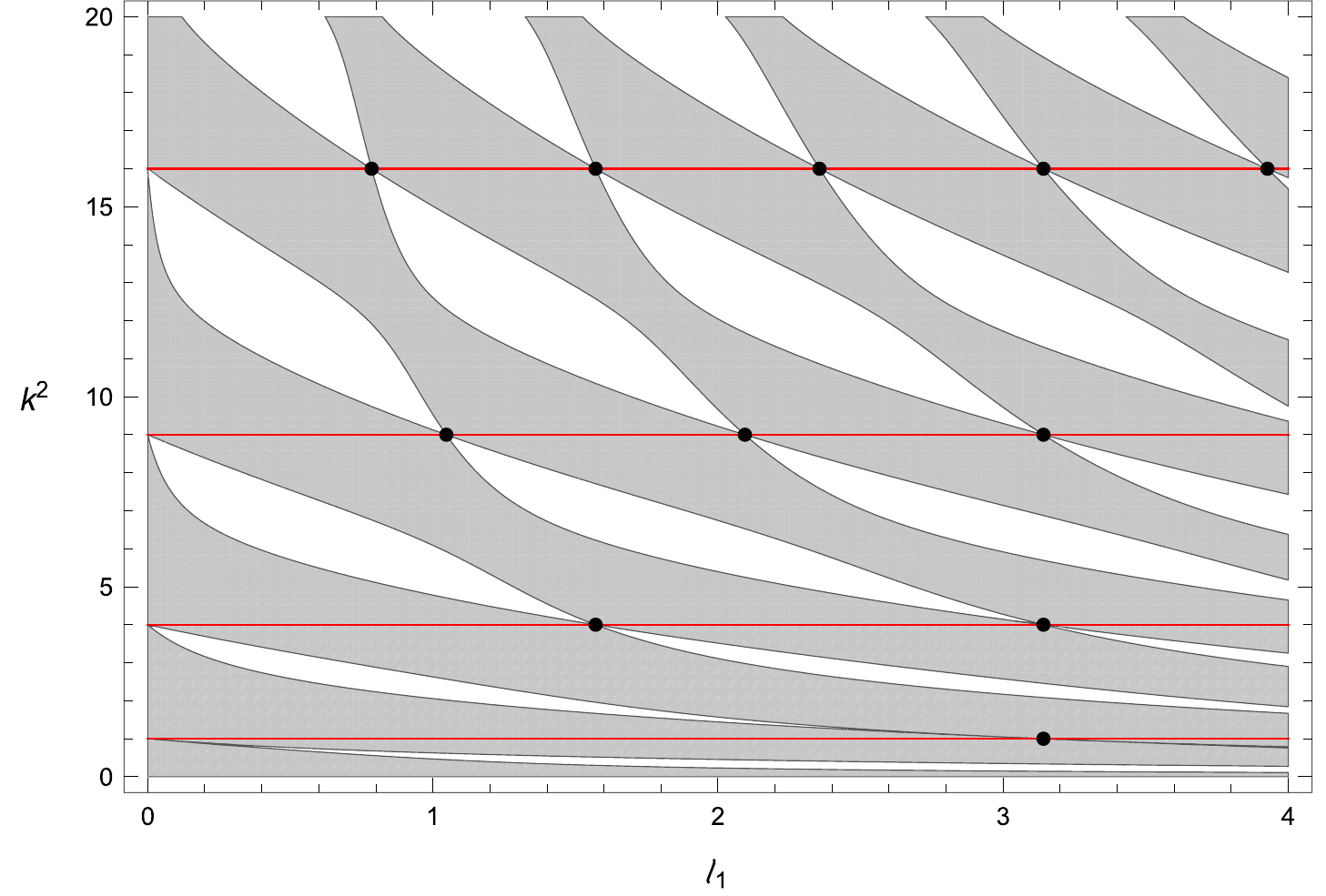}
\caption{The positive spectrum for $\ell=1$. The red lines mark embedded flat bands, the dots the points where the gaps close.}
\label{Figure15}
\end{figure}
In particular, in the high-energy regime, $k\rightarrow\infty$, the number $k^2$ belongs to a band if and only if $\cos 2\pi k \geq \cos 2k(\ell_1+\pi)$, or equivalently
$$ 
\sin k\ell_1 \sin k(\ell_1+2\pi) \ge 0,
$$ 
and to a gap if these quantities are negative, both again modulo an $\OO(k^{-2})$ relative error. The bands and gaps grow in the average at the energy scale; at the momentum scale they are asymptotically periodic if $\ell_1$ is a rational multiple of $\pi$, and aperiodic for $\ell_1$ incommensurate with the ring perimeter. In contrast to the situation discussed in the previous section, there is no `fully conducting' chain here:
\begin{itemize}
\item The probability of belonging to the spectrum is for any $\ell_1>0$ equal to
$$ 
P_{\sigma}(H)= \textstyle{\frac12}\,.
$$ 
\end{itemize}
In the asymptotic regime of large $\ell_1$ the spectrum becomes `dense' in the sense that the number of bands in a fixed interval increases as $\ell_1\to\infty$ but the probability to be in a band remains the same. In contrast to the previous section again, it makes no sense to speak of the limit $\ell_1\to 0$ here.

\subsection{Gap closing}

Asking again under which condition may some gaps close we find that the answer depends on the commensurability of $\ell_1$ with the ring circumference:
\begin{itemize}
\item For $\ell_{1}=\frac{m\pi}{n}$ with coprime $m,n\in\mathbb{N}$, band crossings occur at $k= nj,\: j\in\N$. The argument is similar to that used in Sec.~\ref{cross,sec,ell_1=0}. Denoting the large bracket in \ref{Pos,SC,ell_3=0} by $\mathcal{G}(k,\ell_1,\theta)$, we can write
    the condition for such a crossing as
$$ 
\frac{\partial \mathcal{G}(k,\ell_1,\theta)}{\partial \ell_1}=\frac{\partial \mathcal{G}(k,\ell_1,\theta)}{\partial \theta}=\frac{\partial \mathcal{G}(k,\ell_1,\theta)}{\partial k}=0.
$$ 
Computing the derivatives at $k=nj$, we find that the first two vanish if $\sin(jn(\ell_1+\pi))=0$ and $\sin\theta\,\cos\pi jn=0$, respectively, which is true for for $\ell_{1}=\frac{m\pi}{n}$ and $\theta=0,\pm\pi$. Combining this finding with the remaining condition, we arrive at the equations
$$
2jn\ell^2 (1-\cos\pi jm) \cos\pi jn =0 \quad\text{for}\;\;\; \theta=0
$$
and
$$
-2jn\ell^2 (\cos\pi jm +1) \cos\pi jn =0 \quad\text{for}\;\;\; \theta=\pm\pi.
$$
which are are satisfied for even and odd values of $jm$, respectively.
\item In particular, crossings occur at all the natural values of $k$ if $\ell_{1}=m\pi$.
\item On the other hand, all the gaps remain open if $\ell_{1}$ is not a multiple of $\pi$.
\end{itemize}

\subsection{The negative spectrum}

As in the preceding section, it follows from the general properties of quantum graphs that there is a single negative band. To find it, we have to replace $k$ with $i\kappa$ in \eqref{Pos,SC,ell_3=0} and \eqref{Gap,Pos,ell_3=0}; a number $-\kappa^2$ belongs to the spectrum if
$$ 
G(\kappa):=-4\kappa^2\ell^2 +(\kappa^2\ell^2+1)^2 \cosh 2\pi\kappa -(\kappa^2\ell^2-1)^2 \cosh 2\kappa(\ell_1+\pi)
$$ 
is non-negative. In particular, $-\ell^{-2}$ belongs always to the band which is in this case never flat because $G(\ell^{-1})>0$.

Let us next look how the negative band behaves when $\ell_{1}$ is large. Dividing the spectral condition by $\cosh \kappa (\ell _1+\pi)$, we get
$$ 
(\kappa^2\ell^2-1) \left(\frac{\cos\theta \cosh\pi\kappa}{\cosh\kappa(\ell_1+\pi)}-1\right) +\frac{2 \kappa\ell\,\sin\theta\,\sinh\pi\kappa} {\cosh\kappa(\ell_1+\pi)} = 0
$$ 
which for large $\ell_{1}$ takes the form
$$ 
1-\kappa ^2\ell^2 + 2\big(\cos\theta \cosh\pi\kappa\, (\kappa^2\ell^2-1) +2\kappa \ell\,\sin\theta\, \sinh\pi\kappa\big)\,\mathcal{O}(\e^{-\kappa  \ell_1})=0
$$ 
showing that the band shrinks to the value $-\ell^{-2}$ as $\ell_1\to\infty$. Moreover, putting $\kappa=\ell^{-1}+\delta$ we can solve the resulting equations for $\delta$ arriving thus at the asymptotic expression
$$ 
-\kappa^{2}(\theta) = -\ell^{-2} -4\ell^{-2}\,\e^{-\frac{\pi+\ell_1}{\ell}}\, \sin\theta \sinh\frac{\pi}{\ell} + \mathcal{O}\big(\e^{-\frac{2\ell_1}{\ell}}\big)
$$ 
giving the band width
$$ 
\Delta E = 8\ell^{-2}\,\e^{-\frac{\pi+\ell_1}{\ell}}\, \sinh\frac{\pi}{\ell} + \mathcal{O}\big(\e^{-\frac{2\ell_1}{\ell}}\big).
$$ 
\begin{figure}[!htb]
\centering
\includegraphics[scale=.7]{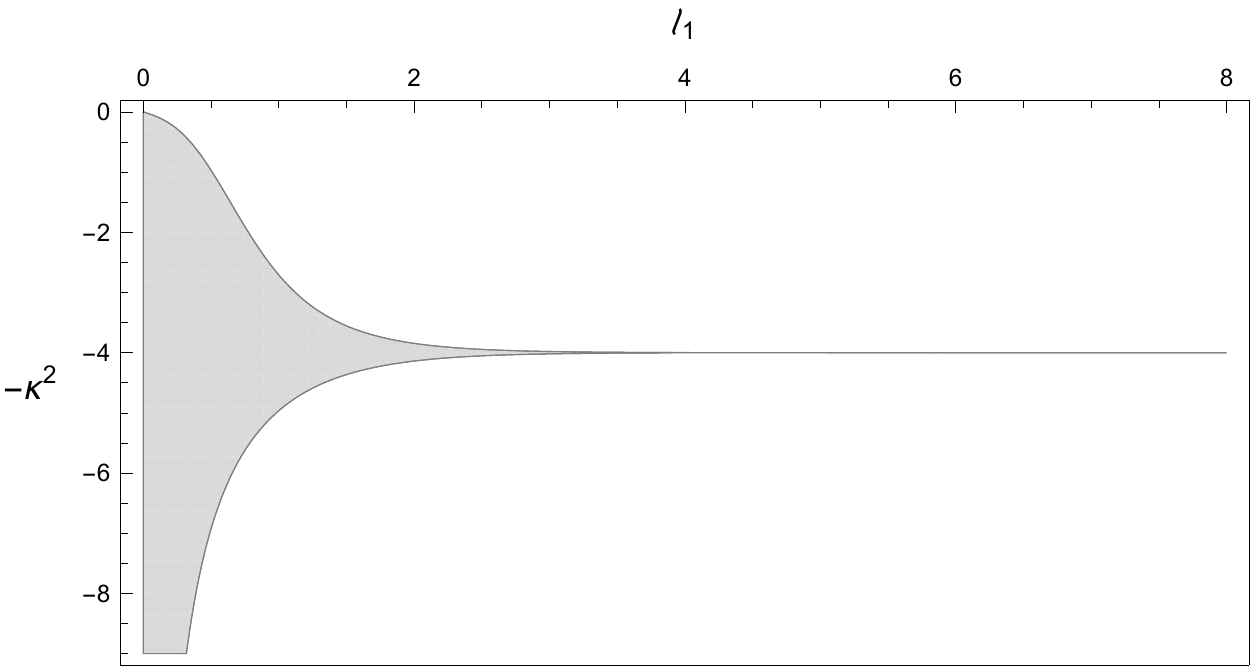}
\caption{The band dependence on $\ell_1$ for $\ell=\frac{1}{2}$.}
\label{Figure16}
\end{figure}
The negative spectrum is shown in Fig.~\ref{Figure16}. It again remains separated from zero since $G(\kappa)=-2\ell_1(2\pi+\ell_1) \kappa^{2}+\mathcal{O}(\kappa^{4})$ is negative for small values of $\kappa$. On the other hand, the band becomes wider as $\ell_1$ decreases but, as is the situation discussed in the previous section, the limit $\ell_1\to 0$ makes no sense.

\bigskip

\subsection*{Acknowledgements}
The research was supported by the Czech Science Foundation project 21-07129S and by the EU project CZ.02.1.01/0.0/0.0/16\textunderscore 019/0000778.

\end{document}